\documentclass{article}
\usepackage{booktabs} % For formal tables
\usepackage[ruled]{algorithm2e} % For algorithms
\usepackage{csquotes}
\usepackage{thmtools} 
\usepackage{thm-restate}
\usepackage[shortlabels]{enumitem}
\usepackage{amsmath}
\usepackage{amsthm}
\usepackage{amssymb}
\usepackage[colorlinks=true, urlcolor=blue]{hyperref}
\usepackage{cleveref}
\usepackage{fullpage}
\usepackage[square,numbers]{natbib}
\usepackage{graphicx}
\usepackage{tikz}
\usepackage{authblk}

\usetikzlibrary{decorations.pathreplacing}

\theoremstyle{theorem}

\newtheorem{lemma}{Lemma}
\newtheorem{theorem}{Theorem}
\newtheorem{proposition}{Proposition}

\theoremstyle{definition}
\newtheorem{definition}{Definition}

\title{Exclusion Zones of Instant Runoff Voting}

\author[1]{Kiran Tomlinson\thanks{kitomlinson@microsoft.com}}
\author[2]{Johan Ugander}
\author[3]{Jon Kleinberg}
\affil[1]{Microsoft Research}
\affil[2]{Stanford University}
\affil[3]{Cornell University}

\date{}

\begin{document}

\maketitle

\begin{abstract}
Recent research on instant runoff voting (IRV) shows that it exhibits a striking combinatorial property in one-dimensional preference spaces: there is an {\em exclusion zone} around the median voter such that if a candidate from the exclusion zone is on the ballot, then the winner must come from the exclusion zone. Thus, in one dimension, IRV cannot elect an extreme candidate as long as a sufficiently moderate candidate is running. In this work, we examine the mathematical structure of exclusion zones as a broad phenomenon in more general preference spaces. We prove that with voters uniformly distributed over any $d$-dimensional hyperrectangle (for $d > 1$), IRV has no nontrivial exclusion zone. However, we also show that IRV exclusion zones are not solely a one-dimensional phenomenon. For irregular higher-dimensional preference spaces with fewer symmetries than hyperrectangles, IRV can exhibit nontrivial exclusion zones. As a further exploration, we study IRV exclusion zones in graph voting, where nodes represent voters who prefer candidates closer to them in the graph. Here, we show that IRV exclusion zones present a surprising computational challenge: even checking whether a given set of positions is an IRV exclusion zone is NP-hard. We develop an efficient randomized approximation algorithm for checking and finding exclusion zones. We also report on computational experiments with exclusion zones in two directions: (i) applying our approximation algorithm to a collection of real-world school friendship networks, we find that about 60\% of these networks have probable nontrivial IRV exclusion zones; and (ii) performing an exhaustive computer search of small graphs and trees, we also find nontrivial IRV exclusion zones in most graphs. While our focus is on IRV, the properties of exclusion zones we establish provide a novel method for analyzing voting systems in metric spaces more generally.
\end{abstract}

% Paper body
\section{Introduction}

An important principle in voting theory is that different voting systems can implicitly favor different regions of the underlying ideological space, for example by benefitting more moderate or more extreme positions. The classic result in this vein is the Median Voter Theorem~\cite{black1948rationale}, which states that when preferences are single-peaked, voting systems satisfying the Condorcet criterion elect the candidate closest to the median voter. That is, Condorcet methods strongly favor moderate candidates in one-dimensional preference spaces. However, this result does not provide us with tools to understand non-Condorcet methods or voting systems under higher-dimensional preferences. In those cases, prior work has often turned to simulation to explore the typical behavior of various voting systems~\cite{elkind2017multiwinner,mcgann2002party,merrill1984comparison}. 

Recent work on instant runoff voting (IRV), which is not a Condorcet method, provides a new approach to theoretically characterizing the candidate positions favored by a voting system. Tomlinson, Ugander, and Kleinberg~\cite{tomlinson2024moderating} proved that IRV with symmetrically-distributed one-dimensional preferences exhibits a striking combinatorial property that is not apparent from its definition: it can only elect a candidate in an interval around the median voter, termed the \emph{exclusion zone}. More generally, an exclusion zone for a given voting system and voter distribution is a set $S$ of candidate positions with the property that if any candidate belongs to $S$, then the winner of the election must belong to $S$. One of the main results of Tomlinson et al.~\cite{tomlinson2024moderating} is that with voters uniformly distributed over $[0, 1]$, the interval $[1/6, 5/6]$ is an exclusion zone of IRV. Thus, treating as usual the more central points of the interval as the more moderate positions, IRV cannot elect an extreme candidate over a sufficiently moderate one~\cite[Theorem 1]{tomlinson2024moderating}.

Our earlier work~\cite{tomlinson2024moderating} left open the question of whether exclusion zones were specific to the structure of the unit interval, or whether they represented a more general principle.  In this work we address this question, extending the analysis of IRV exclusion zones to voters uniformly distributed over metric spaces beyond the unit interval, and fleshing out an understanding of exclusion zones as a general phenomenon of voting systems. We resolve the natural open question about whether the one-dimensional IRV exclusion zone generalizes to higher dimensions, showing for both $L_1$ and $L_2$ distance metrics that uniform voters over any $d$-dimensional hyperrectangle with $d >1$ produce no (nontrivial) IRV exclusion zones. While this strong nonexistence result seems to suggest that IRV exclusion zones are a purely one-dimensional phenomenon, we show this is not the case via a counterexample, an irregular higher-dimensional preference space with a nontrivial IRV exclusion zone. In order to prove our main result for hyperrectangles, we derive general properties of exclusion zones that reveal their structure. We show that the exclusion zones of any voting system in any metric preference space form a nesting from the trivial exclusion zone containing all positions down to a unique \emph{minimal exclusion zone}. These properties of exclusion zones allow us to develop a constructive proof recipe for the nonexistence of nontrivial exclusion zones, which we term the Condorcet Chain Lemma. Informally, this lemma states that if we can construct a sequence of elections such that (1) the first election is lost by a Condorcet winner, (2) the last election is won by a Condorcet loser, and (3) the winner of each election loses the next one, then there is no nontrivial exclusion zone. Our proof for hyperrectangles constructs such a sequence of elections.

In addition to higher-dimensional spatial preferences, we also consider a setting motivated by friendship- or allegiance-based preferences: voting on graphs, where nodes represent voters (some of which run for election) who prefer candidates closer to them in the graph. (Note that the analysis of one-dimensional voting and its implication for exclusion zones has a close analogue in the setting of voting on graphs, by considering elections on path graphs.)  In this setting, we consider the algorithmic problem of identifying IRV exclusion zones in general graphs, which we show is computationally hard. We find that even checking whether a given set of nodes is an IRV exclusion zone is co-NP-complete, and finding a graph's minimal IRV exclusion zone is at least as hard. 

On the positive side, we are able to obtain tractability results by introducing a notion of approximate exclusion zones; roughly speaking, a set of nodes $S$ is an \emph{approximate exclusion zone} (for some approximation parameter $\epsilon$) if a random set of candidates containing at least one element of $S$ has its winner from $S$ with probability at least $1 - \epsilon$.   Using this definition, we are able to give a randomized approximation with the following two-sided guarantee: it outputs a set of nodes that is both guaranteed to be a subset of the true minimal IRV exclusion zone, and also with high probability to be an approximate exclusion zone.  (For this algorithm, we must also design an efficient subroutine for checking if a set of nodes is an approximate exclusion zone.)  This algorithm is useful---both theoretically and in practice---as a certifier that a graph has no nontrivial exclusion zone: its underlying guarantee tells us that if it outputs the full node set, then it has provided such a certificate. We also show that the problem of checking IRV exclusion zones is fixed-parameter tractable in the number of nodes outside the candidate set, so large IRV exclusion zones can be efficiently checked. 

While identifying graph IRV exclusion zones is hard in general, we derive the minimal IRV exclusion zones for several families of graphs, including distance-regular graphs~\cite{van2016distance}, paths, perfect binary trees, and bistars. As in the earlier case of continuous metric spaces, this analysis is aided by the Condorcet Chain Lemma. We also demonstrate the usefulness of our approximation algorithm on a collection of 56 real-world school friendship networks~\cite{paluck2016changing} with several hundred nodes each, much too large for exact computation. We find that 33 of these 56 networks (59\%) likely have approximate IRV exclusion zones, and certify that the other 23 have no nontrivial IRV exclusion zones. The approximate exclusion zones we find tend to be large, comprising an average of 95\% of the network. Informally, nodes that are excluded tend to be on the fringe of the network, although sometimes entire communities are excluded. Through an exhaustive computer search of small graphs and trees, we find that most connected graphs have nontrivial IRV exclusion zones. 

Overall, our work highlights exclusion zones as a rich mathematical object that can illuminate the behavior of voting systems with metric voter preferences. We focus on IRV, partly because it is the predominant single-winner alternative to plurality voting---used for instance in Australia, Ireland, Alaska, and Maine, as well as cities including San Francisco, New York City, and Minneapolis---and partly to resolve open questions from prior work introducing exclusion zones. However, our exploration of properties of exclusion zones applies more broadly, situating them as a general phenomenon for voting, and setting the stage for future work on the exclusion zones of other voting systems.

\section{Related work}
Our work builds directly on the concept of exclusion zones introduced by \citet{tomlinson2024moderating}, where we studied them for IRV and plurality with one-dimensional Euclidean preferences. The concepts of Condorcet winning sets and $\theta$-winning sets~\cite{elkind2011choosing,elkind2015condorcet,bloks2018condorcet,charikar2024six} are related in spirit to exclusion zones, in that they describe sets of candidates preferred by voters to candidates outside the set. However, there are a few important distinctions. First, exclusion zones are sets of possible candidate positions across a family of profiles (the voting theoretic term for a collection of voter preferences over a set of candidates) rather than candidates in a specific profile. As a second matter, exclusion zones are defined by the outcome of the election under a given voting system rather than pairwise preferences of voters.

Another notable approach to studying voting systems in metric spaces comes from the literature on utilitarian metric distortion~\cite{procaccia2006distortion,anshelevich2015approximating,anshelevich2018approximating}. In this framework, voters and candidates have unknown positions in a metric space and the \emph{distortion} of a voting system is the worst-case ratio (over positions consistent with expressed voter rankings) between the total distance from voters to the elected candidate and the minimal total distance to any candidate. No voting system can have metric distortion better than 3~\cite{anshelevich2015approximating} and voting systems achieving this bound are known~\cite{gkatzelis2020resolving,kizilkay2022apluralityveto}. The distortions of other voting rules are also known, including for Borda count, plurality, IRV, and Copeland~\cite{anshelevich2018approximating,anagnostides2022dimensionality}. Utilitarian distortion is a valuable tool for comparing voting systems, but answers a different question than our work. We ask what regions of a space are favored by a voting system over all possible candidate sets rather than how bad an outcome can be in a worst case over unknown voter and candidate positions.  

Relating to our study of voting on graphs, graph-based preferences have a long history in the facility location literature, where facilities can be viewed as candidates and customers can be viewed as voters. A Condorcet node is then a facility preferred by more than half of customers to any other facility~\cite{wendell1981new,bandelt1985networks,hansen1986equivalence}. Graph-based preferences also occasionally appear in the social choice literature as a special case of metric preferences~\cite{skowron2017social}. One working paper uses graph-based voting to explain the success of the Medici family in medieval Florence~\cite{telek2016power}. From a different angle, graph-distance voting has recently been proposed as a node centrality measure~\cite{brandes2022voting,skibski2023closeness}. 

\section{Voting preliminaries}
We begin by defining the terms we will need for our analysis. An \emph{election} consists of a set of candidates $C$, voters $V$ with preferences over $C$, and a voting rule $r$ specifying how a winner should be elected from $C$ based on the voter preferences. We focus on metric voter preferences (also called spatial preferences), where voters and candidates have positions in a bounded metric space $(M,d)$ with set $M$ and metric $d$. Voters prefer candidates closer to them according to $d$. To model a large voting population, we let $V$ be a continuous distribution over $M$ and measure votes with real-valued shares, although similar results would also hold for discrete sets of voters. In this work, we always have $V$ uniformly distributed over $M$, for simplicity rather than necessity. The tuple $(M, d, V, r)$ defines an \emph{election setting}, which in combination with a candidate configuration $C$ fully specifies an election. 

The simplest voting rule is \emph{plurality}, where the candidate who is most preferred by the greatest number of voters is elected (with ties broken in some pre-specified way, say at random). We focus on a prominent alternative to plurality, \emph{instant runoff voting} (IRV), also known as ranked-choice voting, single transferable vote, or the Hare method. Under IRV, we repeatedly compute plurality vote shares for each candidate and eliminate the candidate with the fewest votes, until only the winner remains. In practice, this procedure is accomplished by asking voters to rank the candidates in order of preference, so that these ``runoffs'' can be computed ``instantly.''

A few other concepts from voting theory will be useful to us, with slight modifications for our setting. Given $(M, d, V)$, a \emph{Condorcet winner} is a candidate position $x\in M$ such that for any other $y \in M$, $x$ is preferred to $y$ by at least half of the voters (such a position may not exist). A \emph{Condorcet loser} is a position  $y\in M$ such that for any other $x \in M$, $x$ is preferred to  $y$ by at least half of the voters.  For example, with uniformly distributed voters over $[0, 1]$, the position $1/2$ is a Condorcet winner, while the positions $0$ and $1$ are Condorcet losers. Note that this definition is a slight deviation from the standard one, which defines Condorcet winners and losers as candidates in a particular election. Here, we consider a \emph{position} to be a Condorcet winner if it is preferred by a majority to any alternative in \emph{any} election generated by $(M, d, V)$ (and likewise for losers). A Condorcet voting rule elects a Condorcet winner whenever one is present in the election. Neither IRV nor plurality are Condorcet.

\section{Exclusion zones and their properties}
We now define our object of focus: exclusion zones, first formalized in our prior work~\cite{tomlinson2024moderating}. Intuitively, an exclusion zone of an election setting is a set $S$ of candidate positions such that any election with at least one candidate from $S$ is won by a candidate in $S$. This provide a sense of which positions are favored by the voting rule in the given preference space.

\begin{definition}
  A nonempty set $S\subseteq M$ is an \emph{exclusion zone} of an election setting  $(M, d, V, r)$ if the winner under $r$ with voters $V$ is guaranteed to be in $S$ for any candidate configuration $C$ with $C \cap S \ne \emptyset$, regardless of how ties are broken. 
\end{definition}

When $M = [0, 1]$, $d$ is the Euclidean metric (i.e., voters have \emph{1-Euclidean preferences}~\cite{elkind2022preference}), and the voters $V$ are uniform over $[0, 1]$,  $S = [1/6, 5/6]$ is an exclusion zone of IRV~\cite{tomlinson2024moderating}. IRV also has exclusion zones for other symmetric voter distributions over $[0, 1]$~\cite{tomlinson2024moderating}. However, this earlier work did not further explore the concept of exclusion zones, which we now show exhibit some elegant structure. We begin by showing that exclusion zones are nested inside $M$. 

\begin{proposition}\label{prop:nested}
    Let $S \ne T$ be two exclusion zones of $(M, d, V, r)$. Either $S \subset T$ or $T \subset S$. 
\end{proposition}
\begin{proof}
    By contradiction. Suppose that $S \not\subset T$ and $T \not\subset S$. Then there is some position $s \in S$ such that $s \notin T$ and some $t \in T$ such that $t \notin S$. Consider the election with candidate set $C = \{s, t\}$. If $s$ is not guaranteed to win, then $S$ is not an exclusion zone. If $t$ is not guaranteed to win, then $T$ is not an exclusion zone. So either $S$ or $T$ are not exclusion zones, a contradiction.
\end{proof}

Indeed, in the 1-Euclidean uniform voters case, our prior work~\cite{tomlinson2024moderating} showed that any interval $[c, 1-c]$ for $c \in [0, 1/6]$ is an exclusion zone of IRV, exactly such a nesting, and that no interval smaller than $[1/6, 5/6]$ is an exclusion zone. Here, we formalize the notion of the smallest exclusion zone of a voting system.

\begin{definition}
  An exclusion zone $S$ is \emph{minimal} if no proper subset of $S$ is an exclusion zone.
\end{definition}

In our language, $[1/6, 5/6]$ is the minimal IRV exclusion zone for uniform overs over $[0, 1]$. \Cref{prop:nested} reveals that there cannot be more than one minimal exclusion zone, and we can also show one always exists. 

\begin{proposition}
Let $\mathcal S$ be the set of all exclusion zones of $(M, d, V, r)$. The unique minimal exclusion zone is given by $S^* = \cap_{S \in \mathcal S} S$.
\end{proposition}
\begin{proof}
First, we'll show that $S^*$ is an exclusion zone. Consider any candidate configuration including at least one from $S^*$. Suppose, for the purposes of contradiction, that the winner $w$ does not belong to $S^*$. Since $w \notin S^*$, there exists some $S \in \mathcal S$ such that $w\notin S$. But this configuration includes a candidate from $S$ and the winner is not in $S$, contradicting that $S$ is an exclusion zone. Thus, for any candidate configuration including a candidate from $S^*$, the winner belongs to $S^*$---hence it's an exclusion zone. Next, $S^*$ is clearly minimal, as it is the intersection of all exclusion zones: there cannot be an exclusion zone $T$ with $T \subset S^*$.
\end{proof}

Next, we say that exclusion zone $S = M$ is the \emph{trivial} exclusion zone, which always satisfies the definition. If $S$ is a proper subset of $M$, we call $S$ \emph{nontrivial}. We thus have that the set of exclusion zones forms a nested chain from the trivial exclusion zone $M$ down to the minimal exclusion zone $S^*$. In some cases, $S^* = M$, in which case there is no nontrivial exclusion zone. As a warmup, we now provide some examples of voting systems and their exclusion zones.

\begin{proposition}\label{prop:condorcet-median}
  Let $(M, d)$ be a bounded interval in one dimension with the Euclidean metric. For any Condorcet voting rule $r$ and any voter distribution $V$ over $M$, $\{\text{median}(V)\}$ is the minimal exclusion zone of $(M, d, V, r)$.  
\end{proposition}
\begin{proof}
  By the Median Voter Theorem~\cite{black1948rationale}, any Condorcet rule will always elect the candidate closest to $\text{median}(V)$ with 1-Euclidean preferences. Thus, if a candidate configuration contains $\text{median}(V)$, it will win, satisfying the definition of an exclusion zone.  Any singleton exclusion zone is minimal.
\end{proof}
In this Condorcet setting, we can actually characterize the entire space of exclusion zones, although the statements are simpler if we restrict $M$ to $[0, 1]$ and restrict $V$ to be a symmetric distribution about $1/2$. To do this, we first develop some useful facts about exclusion zones, which will also prove useful for our main results (see Section~\ref{sec:irv_higher_dim}).

\begin{proposition}\label{prop:exclusion-facts}
  Given a election setting $(M, d, V, r)$ such that $r$ picks the majority winner in pairwise contests,
  \begin{enumerate}
    \item[\it(a)] Any Condorcet winner is in the minimal exclusion zone.
  \item[\it(b)]  For any exclusion zone $S$, if there is some candidate configuration including $x\in S$ where $y$ can win, then $y \in S$. 
  \item[\it(c)]  If the minimal exclusion zone contains a Condorcet loser, then the minimal exclusion zone is trivial.
  \end{enumerate}
\end{proposition}
\begin{proof}
 \begin{enumerate}
   \item[(a)]   If not, then consider a pairwise contest between a Condorcet winner and a candidate in the minimal exclusion zone. The Condorcet winner can beat them---contradicting the definition of an exclusion zone.
   \item[(b)]  If not, then we would have found a configuration where a candidate not in $S$ can beat a candidate in $S$, contradicting that $S$ is an exclusion zone.
   \item[(c)]   A Condorcet loser can lose against any candidate $y$ in their pairwise contest, so this follows from (b). \qedhere
 \end{enumerate}
\end{proof}

Using \Cref{prop:exclusion-facts}(b), we can now characterize the space of exclusion zones for Condorcet methods with symmetric 1-Euclidean voters.

\begin{proposition}
  Let $M = [0, 1]$, $d$ be the Euclidean metric, and $V$ be a distribution over $[0, 1]$ symmetric about $1/2$. For any Condorcet voting rule $r$, the set of all exclusion zones of $(M, d, V, r)$ is
  \begin{align*}
    \mathcal S = \{\{1/2\}\} \cup \bigcup_{ c \in [0, 1/2)} \{[c, 1-c], (c, 1-c)\}
  \end{align*}\\
\end{proposition} 
\begin{proof}
  We already saw in \Cref{prop:condorcet-median} that $\{1/2\}$ is the minimal exclusion zone of this election setting and that the candidate closest to $1/2$ will win. Thus, for any symmetric interval around $1/2$ (open or closed), if that interval contains a candidate, then that interval contains the winner, since the candidate closest to $1/2$ is in the interval. Thus all of the sets in $\mathcal S$ are exclusion zones, so it suffices to show no other subset of $[0, 1]$ is an exclusion zone.
  
  Consider any exclusion zone $S$, any $x\in S$, and any $y \in [0, 1]$ that is closer to $1/2$ than $x$. The winner of the election $\{x, y\}$ is $y$, so by \Cref{prop:exclusion-facts}(b), $y \in S$. This implies that only symmetric intervals around $1/2$ can be exclusion zones---otherwise we would have a counterexample configuration. Even half-open intervals $[c, 1-c)$ cannot be exclusion zones, since $1-c$ can win the election $\{c, 1-c\}$ after a tie-breaker.  
\end{proof}

As another example, consider plurality voting with uniformly distributed voters over $[0, 1]$. We already know that in this setting, no proper subset of $(0, 1)$ is an exclusion zone~\cite{tomlinson2024moderating}. We complete this characterization by showing that $(0, 1)$ itself is an exclusion zone, as long as no duplicate candidate positions are allowed (in that case, the minimal exclusion zone of plurality is trivial; throughout the paper, we assume the candidate set has no duplicates, although for IRV this does not affect exclusion zones).

\begin{proposition}
  The minimal exclusion zone of plurality with uniform voters over $[0, 1]$ is $(0, 1)$. 
\end{proposition}
\begin{proof}
  By \cite[Theorem 2]{tomlinson2024moderating}, no proper subset of $(0, 1)$ is an exclusion zone. Suppose we have a candidate configuration with at least one candidate in $(0, 1)$. It suffices to show that the winner cannot be a candidate at 0 or at 1. Suppose there is a candidate at 0 and consider the leftmost candidate $x \in (0, 1)$. This leftmost candidate evenly splits the vote share from voters in $[0, x]$ with the candidate at $0$, but also gets nonzero vote share to the right of $x$ (since there are no duplicate candidates, there is some distance between $x$ and the candidate to its right). Thus $x$ has larger vote share than the candidate at 0, and the candidate at 0 cannot be the plurality winner. A symmetric argument shows a candidate at 1 cannot win.
\end{proof}

In several instances, we will encounter election settings where the minimal exclusion zone is trivial (equivalently, there are no nontrivial exclusion zones), including our main result for hyperrectangles. Proving that a voting system has no nontrivial exclusion zones appears to require proving a broad statement of nonexistence. However, we show that the structure of exclusion zones explored above enables a constructive proof recipe that we term the Condorcet Chain Lemma. Specifically, the facts from \Cref{prop:exclusion-facts} allow us to chain together a sequence of candidate configurations, starting from a Condorcet winner and ending at a Condorcet loser, implying that there are no nontrivial exclusion zones. 
\begin{lemma}[Condorcet Chain Lemma]\label{lemma:chain}
   Given an election setting $(M, d, V, r)$ such that $r$ picks the majority winner in pairwise contests, if there are candidate configurations $C_1, \dots, C_n$ such that:
  \begin{enumerate}
    \item $C_1$ contains a Condorcet winner, but another candidate $w_1$ can win,
    \item each $C_{i+1}$ includes $w_i$, but some other candidate $w_{i+1}$ can win,
    \item $w_n$ is a Condorcet loser,
  \end{enumerate}
  then $(M, d, V, r)$ has no nontrivial exclusion zones.
\end{lemma}
\begin{proof}
  By \Cref{prop:exclusion-facts}(a), the minimal exclusion zone $S^*$ contains all Condorcet winners. But $w_1$ defeats a Condorcet winner, so $w_1 \in S$ by \Cref{prop:exclusion-facts}(b). Repeated applications of \Cref{prop:exclusion-facts}(b) along the chain of configurations (i.e., induction) show that each $w_i \in S^*$, so $w_n \in S^*$. By \Cref{prop:exclusion-facts}(c), this means $S^* = M$, as $w_n$ is a Condorcet loser, so there are no nontrivial exclusion zones. 
\end{proof}

\section{IRV exclusion zones in higher dimensions}
\label{sec:irv_higher_dim}
Now that we have established some useful properties of exclusion zones, we turn to our first main results. While \citet{tomlinson2024moderating} characterized the exclusion zones of IRV in one-dimension, that work left a significant open question: does IRV have nontrivial exclusion zones with higher-dimensional preferences? The proof technique in one dimension fails to generalize, leaving the answer uncertain. Here, we resolve this open question: using the Condorcet Chain Lemma, we show that IRV has no nontrivial exclusion zones with uniformly distributed voters over hyperrectangles of dimension two or greater (with both $L_1$ and $L_2$ distance metrics). However, we also show that there exist higher-dimensional preference structures that induce nontrivial IRV exclusion zones. The example we construct generates an IRV exclusion zone by breaking the strong symmetry of hyperrectangles, which our proof suggests is the cause of their lack of IRV exclusion zones.  

To apply the Condorcet Chain Lemma to higher-dimensional Euclidean preferences, we start by finding Condorcet winners and losers in hyperrectangles. We say that voters have $L_p$ preferences, or are $L_p$ voters, in a $d$-dimensional space if they vote according to the $L_p$ metric, with distance function $d_p(x, y) =\left(\sum_{i = 1}^d |x_i - y_i|^p\right)^{1/p}$.  

\begin{lemma}\label{lemma:hyperrect-condorcet}
  With uniform $L_p$ voters over a $d$-dimensional hyperrectangle $[0, w_1]\times\dots \times [0, w_d]$, the center $c = (w_1/2, \dots, w_d/2)$ is a Condorcet winner and the corners are Condorcet losers for any $d \ge 1$ and $p \ge 1$. 
\end{lemma}
See \Cref{app:proofs} for the proof and for all others that we have omitted or sketched. Given that the center of a hyperrectangle is a Condorcet winner and the corners are Condorcet losers, if we can find a sequence of configurations where the winner of each configuration loses in the next and linking the center and the corners, then that hyperrectangle has no nontrivial exclusion zones by the Condorcet Chain Lemma. This is impossible in one dimension, but we find such a sequence for any $d \ge 2$. As a warmup, we show what one such sequence looks like, for the special case of a square with uniform $L_2$ voters.

\begin{proposition}\label{prop:l2-square}
  The square with uniform $L_2$ voters has no nontrivial IRV exclusion zone.
\end{proposition}

As a sketch of proof, in \Cref{fig:square-sequence} we provide a sequence of configurations that satisfies the requirements of the Condorcet Chain Lemma, where the red candidate loses at each step, but was a possible winner in the previous step (denoted in blue). A full proof is given in \Cref{app:proofs}. 

\begin{figure}[t]
\centering
  \begin{tikzpicture}[x=3mm,y=3mm]

  \draw[draw=black, thick] (0, 0) rectangle ++(5, 5);
    
  \draw[dashed] (1.5, 3.5) -- (1.5, 1.5) -- (3.5, 1.5) -- (3.5, 3.5) -- (1.5, 3.5);
   \draw[dashed] (1.5, 1.5) -- (0, 0);
   \draw[dashed] (3.5, 3.5) -- (5, 5);
   \draw[dashed] (1.5, 3.5) -- (0, 5);
   \draw[dashed] (3.5, 1.5) -- (5, 0);

  \node[fill=red!70, circle,inner sep=2pt]  at (2.5, 2.5) {};
  \node[fill=blue!70, circle,inner sep=2pt] (M) at (2.5, 0.5) {};
    \node[fill, circle,inner sep=2pt]  at (2.5, 4.5) {};
    \node[fill, circle,inner sep=2pt]  at (0.5, 2.5) {};
    \node[fill, circle,inner sep=2pt]  at (4.5, 2.5) {};

\end{tikzpicture}
  \quad\raisebox{6.7mm}{$\rightarrow$}\quad
   \begin{tikzpicture}[x=3mm,y=3mm]
  \draw[draw=black, thick] (0, 0) rectangle ++(5, 5);

\draw[dashed] (0.05, 0) -- (2.5, 3.0625);

\draw[dashed] (4.95, 0) -- (2.5, 3.0625);
  \draw[dashed] (2.5, 5) -- (2.5, 3.0625);
  
  \node[fill=blue!70, circle,inner sep=2pt] (L) at (0, 2.5) {};
  \node[fill=red!70, circle,inner sep=2pt] (M) at (2.5, 0.5) {};
  \node[fill, circle,inner sep=2pt] (R) at (5, 2.5) {};

\end{tikzpicture}
  \quad\raisebox{6.7mm}{$\rightarrow$}\quad
        \begin{tikzpicture}[x=3mm,y=3mm]
  \draw[draw=black, thick] (0, 0) rectangle ++(5, 5);

  \draw[dashed] (0, 3.25) -- (5, 3.25);
  \draw[dashed] (0, 1.75) -- (5, 1.75);

  \node[fill=blue!70, circle,inner sep=2pt] (L) at (0, 4) {};
  \node[fill=red!70, circle,inner sep=2pt] (M) at (0, 2.5) {};
    \node[fill, circle,inner sep=2pt] (R) at (0, 1) {};

\end{tikzpicture} \quad \raisebox{6.7mm}{$\rightarrow$}\quad
 \begin{tikzpicture}[x=3mm,y=3mm]
  \draw[draw=black, thick] (0, 0) rectangle ++(5, 5);

  \draw[dashed] (0, 2) -- (3, 2);
  \draw[dashed] (5, 0) -- (3, 2);
  \draw[dashed] (2.4, 5) -- (3, 2);
  
  \node[fill=blue!70, circle,inner sep=2pt] (B) at (0, 0) {};
  \node[fill=red!70, circle,inner sep=2pt] (M) at (0, 4) {};
  \node[fill, circle,inner sep=2pt] (T) at (5, 5) {};
\end{tikzpicture}
\caption{A visual sketch of our proof of \Cref{prop:l2-square}, showing a sequence of elections satisfying the Condorcet Chain Lemma, proving that it has no nontrivial IRV exclusion zones with uniform $L_2$ voters. In each configuration, the red candidate is eliminated first and the blue candidate is a possible winner of the resulting tiebreak. Critically, the red candidate was a possible winner of the previous configuration. The first configuration includes the center, the Condorcet winner, and the winner of the last configuration is a corner, a Condorcet loser. The exact positions and vote shares are given in the proof in \Cref{app:proofs}.}\label{fig:square-sequence}
\end{figure}
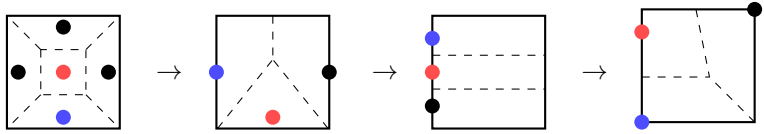

Moving on from the two-dimensional square, we now provide a much more general result, although the sequence of configurations requires many more steps. We begin by showing that it suffices to construct such a sequence for two-dimensional rectangles---we can essentially decompose a hyperrectangle into a sequence of rectangles, one for each dimension, and make progress towards a corner one dimension at a time. 

\begin{lemma}\label{lemma:rect-to-hyper}
  Let $p \ge 1$. If there is a sequence of configurations satisfying \Cref{lemma:chain} for any rectangle with uniform $L_p$ voters, then every $d$-dimensional hyperrectangle for $d \ge 2$ has no nontrivial exclusion zones with uniform $L_p$ voters. 
\end{lemma}
\begin{proof}[Proof sketch]
Consider an election in a hyperrectangle $R$ where all candidates lie on a plane $P$ parallel to a pair of axes and containing the center. We show that this election plays out exactly like an election in the two-dimensional projection of the hyperrectangle onto $P$. Thus, we can use the chain of configurations in the rectangle to show that every point in $P \cap R$ is in the minimal exclusion zone of $R$. We can then repeat the procedure with a different plane $P'$, starting from a point we know is in the minimal exclusion zone of $R$, but which has one coordinate equal to 0. Repeating the rectangle construction $d$ times then shows the origin is in the minimal exclusion zone of $R$, which we know is a Condorcet loser by \Cref{lemma:hyperrect-condorcet}.
\end{proof}

We complete the proof by providing such sequences for arbitrary rectangles with uniform $L_1$ and $L_2$ voters. This yields our main result for higher dimensional preference spaces.

\begin{theorem}\label{thm:hyperrect-trivial-minimal}
Every $d$-dimensional hyperrectangle for $d \ge 2$ has no nontrivial exclusion zones with uniform $L_1$ or $L_2$ voters. 
\end{theorem}
\begin{proof}[Proof sketch]
  We start by considering arbitrary rectangles and find sequences of configurations satisfying the Condorcet Chain Lemma for both $L_1$ and $L_2$ preferences. These are similar in flavor to the sequence given in \Cref{prop:l2-square}, but with numbers of steps growing with the narrowness of the rectangle. Applying \Cref{lemma:rect-to-hyper} then produces a sequence of configurations in the hyperrectangle satisfying the Condorcet Chain Lemma.
\end{proof}

This strong result seems to suggest that nontrivial exclusion zones for IRV are a purely one-dimensional phenomenon. However, we now show this is not the case. Indeed, a second dimension does make it much easier for non-central candidates to combine and squeeze out a more central candidate, but this needs to be paired with the strong symmetry of hyperrectangles with uniform voters in order to make the minimal exclusion zone trivial. By breaking this symmetry, we can find higher-dimensional preference spaces with nontrivial IRV exclusion zones. 

%\begin{theorem}\label{thm:barbell}
%  Consider the shape $B$ formed by linking two squares of side length $1$ and $2$ at their bases by a rectangle of height $1/10$ and width 8, where the smaller square has its lower left corner at the origin:
%  \begin{center}
%    \begin{tikzpicture}[x=6mm, y=6mm]
%      \draw (0, 0) -- (0, 1) -- (1, 1) -- (1, 0.1) -- (9, 0.1) -- (9, 2) -- (11, 2) -- (11, 0) -- cycle;
%      \fill[opacity = 0.3] (5, 0) -- (4.9, 0.1) -- (9, 0.1) -- (9, 2) -- (11, 2) -- (11, 0) -- cycle;
%      \node at (10, 1) {$S$};
%    \end{tikzpicture}
%  \end{center}
%  The shaded set $S = \{(x, y) \in B \mid x + y \ge 5 \}$ is an IRV exclusion zone with uniform $L_1$ voters over $B$. 
%\end{theorem}

\begin{theorem}\label{thm:golf-flag}
  Consider the shape $F$  (see \Cref{fig:funny-shape}) formed by a rectangle of height $1/10$ and width $8$ (having its lower left corner at the origin) with two right triangles of side lengths $(2, 2, \sqrt 8)$ and $(1, 1, \sqrt 2)$ placed on the top left and bottom right of the rectangle. The set $S = \{(x, y) \in F \mid x - y \le 6\}$ is an IRV exclusion zone with uniform $L_1$ voters over $F$. 
\end{theorem}

\begin{figure}[t]
\centering
  \begin{tikzpicture}[x=6mm, y=6mm]
      \draw (0, 0) -- (0, 2.1) -- (2, 0.1) -- (8, 0.1) -- (8, -1) -- (7, 0) -- cycle;
      \fill[opacity = 0.3] (0, 0) -- (0, 2.1) -- (2, 0.1) -- (6.1, 0.1) -- (6, 0) -- cycle;
      \node at (0.6, 0.6) {$S$};
    \end{tikzpicture}
    \caption{The shape $F$ from \Cref{thm:golf-flag}, which has a nontrivial IRV exclusion zone with uniform $L_1$ voters. The shaded set $S$ is a nontrivial IRV exclusion zone.}\label{fig:funny-shape}
\end{figure}
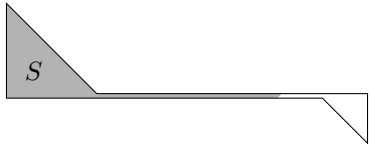

\begin{proof}
 Call the larger triangle $T$. The area of $T$ is $2$, while the area of $F$ is $2 + 1/2 + 8/10 = 3.3$.  
 Suppose only one candidate $c = (c_x, c_y)$ remains in $S$. Consider two cases. 
 
 (1) If $c_x \le 2$ , then $c$ gets the full vote share of $T$, since the largest $L_1$ distance between a point in $T$ and a candidate with $c_x\le 2$ is 4.1 (between $(0, 2.1)$ and $(2, 0)$), while the smallest distance from a point in $T$ (namely, $(2, 0.1)$) to a point outside $S$ is $> 4.1$.
 
 (2) If $c_x > 2$, then consider a shortest path between a voter $v\in T$  and $c$ that first descends to $y = 0.1$, then moves right along the top part of the rectangle, then drops down to $c_y$. A shortest path from $v$ leaving $S$ would be strictly longer, since one such path is achieved by extending this shortest path to $c$ to the right. Therefore every voter in $T$ prefers $c$ to any candidate outside of $S$.
 
 Thus, the last candidate $c$ remaining in $S$ during IRV elimination gets the full vote share in $T$; since this is more than half the total vote share, $c$ wins. $S$ is therefore an IRV exclusion zone. 
 \end{proof}

The preference space $F$ in \Cref{thm:golf-flag} has a natural interpretation: there are two relevant policy dimensions, one in which voters are strongly polarized towards opposite extremes (the $x$ dimension) and one in which voters are mostly moderate (the $y$ dimension), but tend to have opposite leanings on opposite sides of the $x$ dimension.  In this setting, IRV can never elect a candidate on the smaller extreme side. The same idea can be generalized to higher dimensions. However, the construction does not seem to work for $L_2$ preferences; it remains an open question whether there are connected higher-dimensional preference spaces with nontrivial IRV exclusion zones for $L_2$ voters (if we allow disconnected spaces, consider uniform voters over a large polytope and a small polytope that are very far apart: only candidates in the large polytope can win). 

\section{Voting on graphs}
We now turn our attention to a different metric space: unweighted graphs. In this setting, voters and candidates are nodes in a graph, with preferences determined by path distance. We assume every node in the graph represents a single voter and some subset of these nodes run as candidates. Note that distance ties are common in unweighted graphs; we say that each node has vote share 1 that it evenly distributes among all closest candidates.\footnote{This approach to resolving indifference among voters has recently been called Split-IRV to contrast it with Approval-IRV, where each tied candidate receives one full approval vote~\cite{delemazure2024generalizing}. In the context of single transferrable vote (multi-winner IRV), Split-IRV is also called Meek's method~\cite{hill1987algorithm}, after its inventor~\cite{meek1969nouvelle,meek1994new}. We use Split-IRV rather than Approval-IRV as it more closely parallels the continuous metric space case. Split-IRV has also been used in real-world elections~\cite{mollison2023fair}, although there are good theoretical reasons to prefer Approval-IRV~\cite{delemazure2024generalizing}; see~\cite{hill2005meek} for additional debate about approaches to resolving indifference in IRV.} Formally, we define a graph election as follows. 

\begin{definition}
  An election on a graph $G = (V, E)$ has metric space $(V, d_G)$, where $d_G$ is the path distance metric of $G$, candidates $C \subseteq V$, and a uniform voter distribution over $V$.  
\end{definition}

Such graph-distance preferences are common in the literature on facility location~\cite{wendell1981new,bandelt1985networks,hansen1986equivalence}, and in our voting setting can model friendship- or allegiance-based voting~\cite{telek2016power}. For instance, consider a class president election. If every student votes for the candidate they are closest friends with (measured by path distance in the friendship network), then their preferences are exactly given by the graph metric. 

IRV in this graph voting setting can also be motivated from sequential closures in facility location problems. Here, the graph represents the spatial distribution of customers, with facilities acting as candidates. Every customer uses the closest facility. At each round of the process, the facility with the fewest customers shuts down, causing its customers to redistribute themselves to the next-closest facility (just as votes are reallocated in IRV with the same graph metric). Under this process, we can ask where the last remaining facility will be located, given some graph---this is precisely equivalent to asking where the IRV winner will be, given graph-distance preferences. The \emph{facility delocation} problem has been studied from a strategic perspective~\cite{bhaumik2010optimal,ruiz2017cournot}; our study of IRV exclusion zones in graphs gives some additional insight into the purely mechanistic sequential closure process, focusing on sequential closures down to a single location. 

\subsection{Identifying graph IRV exclusion zones in general}\label{sec:graph-general}
Given our graph voting setting, we begin by considering a general algorithmic problem: can we identify the IRV exclusion zones of a given graph?

\begin{definition}
  Given a graph $G = (V, E)$ and a set of nodes $S \subseteq V$, \textsc{IRV-Exclusion} is the decision problem asking whether $S$ is an IRV exclusion zone of $G$. \textsc{Min-IRV-Exclusion} is the optimization problem whose solution is the minimal IRV exclusion zone of $G$. 
\end{definition}

We show that exclusion zones are very difficult to identify in general---even checking whether a given set is an IRV exclusion zone is computationally hard. For a full proof, see \Cref{app:proofs}. 

\begin{theorem}\label{thm:exclusion-co-np}
\textsc{IRV-Exclusion} is co-NP-complete. 
\end{theorem}

\begin{proof}[Proof sketch]
By reduction from restricted exact cover by 3-sets (RX3C), where each item appears in exactly 3 sets, which is NP-hard~\cite{gonzalez1985clustering}. We construct a graph from the RX3C instance where a specific two-node set is an IRV exclusion zone if and only if there is no exact cover; the only candidate configuration where one of the two nodes can be eliminated is exactly when there are candidates at the nodes corresponding to the exact cover. Inclusion in co-NP follows from a verifier that provides a candidate configuration showing the set is not an exclusion zone.
\end{proof}

Since checking whether even a single set is an exclusion zone is hard, the minimization problem of finding the minimal exclusion zone is naturally also hard.

\begin{theorem}
  \textsc{Min-IRV-Exclusion} is NP-hard.
\end{theorem}
\begin{proof}
If there was a polynomial-time algorithm for  \textsc{Min-IRV-Exclusion}, then we could apply it to the graph constructed in the proof of \Cref{thm:exclusion-co-np}. If it returns the two-node set, then we know there is no exact cover, and we would have a polynomial-time algorithm for \textsc{RX3C}. (No other two-node set can be an exclusion zone of this graph, as those two nodes are the only Condorcet winners.)
\end{proof}

Despite the hardness of identifying exclusion zones, we show that we can identify \emph{approximate} exclusion zones, which we define to be node sets that behave like exclusion zones most of the time.

\begin{definition}
  A set of nodes $S$ is a $(1-\epsilon)$-\emph{approximate exclusion zone} if, drawing a uniformly random candidate configuration $C \subseteq V$ with $C \cap S \ne \emptyset$, the winner is in $S$ w.p.\ at least $1- \epsilon$.  
\end{definition}

\begin{theorem}\label{thm:approx}
  Let $G$ be a graph with $n$ nodes and $m$ edges, and pick any desired  $\epsilon,\delta \in (0, 1)$. There is a randomized algorithm returning a set $S$ in time $O((n^3 +n^2m)\log(1/\delta)/\epsilon^2)$ such that:
  \begin{enumerate}
    \item $S$ is a subset of the minimal IRV exclusion zone of $G$, and
    \item $S$ is a $(1-\epsilon)$-approximate IRV exclusion zone with probability at least $1-\delta$.
  \end{enumerate}
\end{theorem}

\begin{proof}[Proof sketch]
The idea behind the algorithm is to grow a set of nodes which must all be in the minimal exclusion zone:
\begin{enumerate}
  \item Build the pairwise loss graph $L$: for every pairwise contest where $u$ loses or draws against $v$, add the directed edge $(u, v)$. Let $d_L(u)$ be the set of nodes reachable from $u$ in $L$. 
  \item Place a candidate at every node, and run IRV $n$ times with random tiebreaks. Let $S$ be the set of winners in these $n$ elections and for each $u \in S$, add $d_L(u)$ to $S$. 
  \item Repeat until we go $\lceil\log(2/\delta)/(2\epsilon^2)\rceil$ iterations with no updates to $S$:
  \begin{enumerate}
    \item Build a node set $X$ by uniformly sampling one node $u$ in $S$ and then uniformly sampling a subset of $V\setminus \{u\}$.  
    \item Run IRV on $X$. If the winner $w$ is not in $S$, add $w$ and $d_L(w)$ to $S$.
  \end{enumerate} 
\end{enumerate}
(The pairwise loss graph is not needed for correctness, but dramatically speeds up convergence in practice.)
We use \Cref{prop:exclusion-facts} to show that every node added to $S$ must be in the minimal exclusion zone and Hoeffding's inequality to show that the set we end up with is an approximate exclusion zone with high probability.
\end{proof}

We can also show that the computational hardness arises specifically from checking whether small node sets are exclusion zones. Intuitively, for small node sets, there are exponentially many candidate configurations (in the number of nodes outside the set) which could be counterexamples. Formally, we show that \textsc{IRV-Exclusion} is fixed-parameter tractable for the parameter $|V\setminus S|$.

\begin{theorem}\label{thm:fpt}
Let $G$ be a graph with $n$ nodes and $m$ edges. For $|S| = n - c$, there is an algorithm for \textsc{IRV-Exclusion} with runtime $O(2^c n(n+m))$.
\end{theorem}
\begin{proof}
  Consider the following algorithm:
  \begin{enumerate}
    \item For each $u \in S$:
    \begin{enumerate}
    \item For every configuration $X$ of candidates from $\overline S$:
    \begin{enumerate}
    \item If $u$ has the smallest plurality vote share against $X$, then return FALSE. 
    \end{enumerate}
    \end{enumerate}
    \item Return TRUE.
  \end{enumerate}

We begin with correctness. Suppose $S$ is an IRV exclusion zone. Then $u$ can never have the smallest plurality vote share against any configuration $X$ from $\overline S$, so the algorithm returns TRUE. If $S$ is not an IRV exclusion zone, then there is some configuration of candidates with candidates in $S$ where a candidate outside $S$ wins. For this counterexample configuration, at some point, the last candidate in $S$ is eliminated with the smallest plurality vote share; the algorithm checks this stage of the configuration at some point and returns FALSE.
  
  Computing plurality vote shares takes time $O(n+m)$ using a single BFS pass. There are $O(n)$ nodes in $S$ and $2^c$ configurations of candidates from  $\overline S$ to check. This gives the claimed runtime.
  \end{proof}
  
For \textsc{Min-IRV-Exclusion}, we can test progressively larger node sets with the algorithm from \Cref{thm:fpt}. We can speed this up dramatically in practice by ruling out node sets which cannot be exclusion zones using the pairwise loss graph from \Cref{thm:approx}: to be an exclusion zone, all nodes reachable from the set in the pairwise loss graph must also be in the set. We used this algorithm to find the minimal IRV exclusion zones of all  connected graphs on 3--7 nodes and all trees on 3--15 nodes. Some examples are shown in \Cref{fig:graph-examples}. In \Cref{tab:graphs,tab:trees}, we show the number of graphs and trees with nontrivial and 2-node IRV exclusion zones. (No graph can have a one-node exclusion zone, since any one node can be the first eliminated by a tiebreak when there is a candidate at every node.) From these tables we can conclude that the vast majority of small graphs and trees have nontrivial exclusion zones, indicating that it is common for small graphs to have sets of nodes that are easily excluded in IRV graph voting, meaning that if any single node in the exclusion zone were a candidate in the race, the excluded nodes would lose.

\begin{figure}
\centering
  \begin{tikzpicture}
  % 4-cycle
  \begin{scope}[xshift=0cm,scale=0.5]
    \foreach \i in {1,2,3,4}
      \node[circle, draw, fill=blue!70] (\i) at (90*\i:1) {};
    \foreach \i/\j in {1/2,2/3,3/4,4/1}
      \draw (\i) -- (\j);
      
    \node[align=center] at (0, -2) {\small (a)};
  \end{scope}

  % Length-5 path
  \begin{scope}[xshift=1.5cm, scale=0.5]
   \node[circle, draw, fill=red!70] (n1) at (0, 0) {};
    \node[circle, draw, fill=blue!70] (n2) at (1, 0) {};
    \node[circle, draw, fill=blue!70] (n3) at (2, 0) {};
    \node[circle, draw, fill=blue!70] (n4) at (3, 0) {};
    \node[circle, draw, fill=blue!70] (n5) at (4, 0) {};
    \node[circle, draw, fill=red!70] (n6) at (5, 0) {};
    \draw (n1) -- (n2) -- (n3) -- (n4) -- (n5) -- (n6);
    
        \node[align=center] at (2.5, -2) {\small (b)};
  \end{scope}

  % Bistar with 3 leaves on each side
  \begin{scope}[xshift=6cm, scale=0.5]
    \node[circle, draw, fill=blue!70] (center1) at (-1, 0) {};
    \node[circle, draw, fill=blue!70] (center2) at (0, 0) {};
    \foreach \i in {1,2,3}
      \node[circle, draw, fill=red!70] (l\i) at (-2, \i-2) {};
    \foreach \i in {1,2,3}
      \node[circle, draw, fill=red!70] (r\i) at (1, \i-2) {};
    \foreach \i in {1,2,3} {
      \draw (center1) -- (l\i);
      \draw (center2) -- (r\i);
    }
    \draw (center1) -- (center2);
    
      \node[align=center] at (-0.5, -2) {\small (c)};
  \end{scope}

  % Height-2 perfect binary tree
  \begin{scope}[xshift=8.3cm,scale=0.5]
    \node[circle, draw, fill=blue!70] (root) at (0, 1) {};
    \node[circle, draw, fill=blue!70] (left) at (-1, 0) {};
    \node[circle, draw, fill=blue!70] (right) at (1, 0) {};
    \node[circle, draw, fill=red!70] (leftleft) at (-1.5, -1) {};
    \node[circle, draw, fill=red!70] (leftright) at (-0.5, -1) {};
    \node[circle, draw, fill=red!70] (rightleft) at (0.5, -1) {};
    \node[circle, draw, fill=red!70] (rightright) at (1.5, -1) {};
    \foreach \i/\j in {root/left, root/right, left/leftleft, left/leftright, right/rightleft, right/rightright}
      \draw (\i) -- (\j);
      
          \node[align=center] at (0, -2) {\small (d)};
  \end{scope}
  
  \begin{scope}[xshift=10.3cm, scale=0.5]
  \node[circle, draw, fill=blue!70] (A) at (0, 0) {};
  \node[circle, draw, fill=blue!70] (B) at (-0.6, -1) {};
  \node[circle, draw, fill=blue!70] (C) at (0.6, -1) {};
  \node[circle, draw, fill=red!70] (D) at (0, 1) {};

  \draw (A) -- (B);
  \draw (A) -- (C);
  \draw (B) -- (C);
  \draw (A) -- (D);  
  
      \node[align=center] at (0, -2) {\small (e)};
  \end{scope}

  \begin{scope}[xshift=12cm, scale=0.5]
  \node[circle, draw, fill=blue!70] (A) at (-0.6, 0) {};
  \node[circle, draw, fill=blue!70] (B) at (0.6, 0) {};
  \node[circle, draw, fill=red!70] (C) at (0, -1) {};
  \node[circle, draw, fill=red!70] (D) at (-1.2, 1) {};
  \node[circle, draw, fill=red!70] (E) at (1.2, 1) {};

  \draw (A) -- (B);
  \draw (A) -- (C);
  \draw (B) -- (C);
  \draw (A) -- (D);  
    \draw (B) -- (E);  
    
        \node[align=center] at (0, -2) {\small (f)};
  \end{scope}

\end{tikzpicture}
\caption{Some graphs with their minimal IRV exclusion zones in blue and excluded nodes in red: (a) the 4-cycle, (b) the 6-path, (c) the 6-leaf bistar, (d) the height-2 perfect binary tree, (e) the smallest connected cyclic graph with a nontrivial IRV exclusion zone, and (f) the smallest (in nodes, then in edges) connected graph whose minimal IRV exclusion zone does not consist of all non-leaf nodes. }\label{fig:graph-examples}
\end{figure}
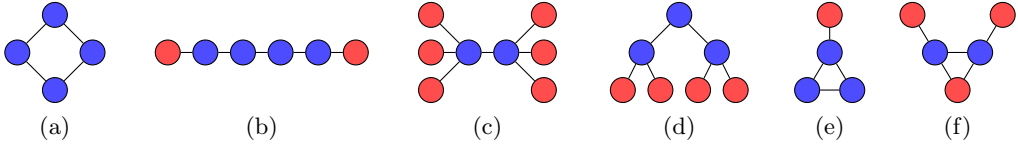

\begin{table}
\centering
  \caption{Number of connected graphs with nontrivial and 2-node IRV exclusion zones for small $n$.}\label{tab:graphs}
    \vspace{0.5em}
  \begin{tabular}{lrrr}
  \toprule
    $n$ & \# Connected Graphs & \# Nontrivial & \# 2-node\\
    \midrule
$3$ & $2$ & $0$ & $0$\\
$4$ & $6$ & $2$ & $1$\\
$5$ & $21$ & $12$ & $2$\\
$6$ & $112$ & $80$ & $10$\\
$7$ & $853$ & $712$ & $56$\\
\bottomrule
  \end{tabular}
\end{table}

\begin{table}
\centering
  \caption{Number of trees with nontrivial and 2-node IRV exclusion zones for small $n$.}\label{tab:trees}
  \vspace{0.5em}
  \begin{tabular}{lrrr}
  \toprule
    $n$ & \# Trees & \# Nontrivial & \# 2-node\\
    \midrule
$3$ & $1$ & $0$ & $0$\\
$4$ & $2$ & $1$ & $1$\\
$5$ & $3$ & $2$ & $1$\\
$6$ & $6$ & $5$ & $2$\\
$7$ & $11$ & $10$ & $3$\\
$8$ & $23$ & $22$ & $6$\\
$9$ & $47$ & $45$ & $10$\\
$10$ & $106$ & $101$ & $21$\\
$11$ & $235$ & $229$ & $41$\\
$12$ & $551$ & $532$ & $73$\\
$13$ & $1301$ & $1280$ & $160$\\
$14$ & $3159$ & $3062$ & $330$\\
$15$ & $7741$ & $7572$ & $672$\\
\bottomrule
  \end{tabular}
\end{table}

\subsection{Families of graphs with known IRV exclusion zones}

While finding exact IRV exclusion zones in graphs is hard in general, we can identify them in some families of larger graphs.  First, in any graph where every pairwise contest is a tie (i.e., every node is a weak Condorcet winner), the minimal exclusion zone is trivial by \Cref{prop:exclusion-facts}. Such graphs include complete graphs, cycles, and all other distance-regular graphs~\cite{van2016distance}. Meanwhile we are also able to show examples of graphs with nontrivial IRV exclusion zones: paths, bistars, and even-height perfect binary trees. 

The bistar graph, consisting of two star graphs of equal size whose centers are joined by an edge, is the simplest example of a graph with the smallest possible IRV exclusion zone: two nodes. 

\begin{proposition}
  The minimal IRV exclusion zone of a bistar graph consists of its two hub nodes.
\end{proposition}
\begin{proof}
Consider the point when only one candidate in the hub nodes remains during IRV elimination. Orient the graph so that this candidate occupies the left hub node. If the leaves on one side are all unoccupied, then the hub candidate has vote share greater than $n/2$ and wins. So suppose there is at least one candidate remaining in the leaves on both sides. Any candidate in a left leaf has a strictly smaller vote share than the hub candidate, since they equally split the available vote share in unoccupied left leaves, but the hub node also gets votes from the other unoccupied hub. Thus the hub candidate is not eliminated next; this continues to be true until all left leaves are unoccupied, at which point the hub candidate wins.  
\end{proof}

The result of \citet{tomlinson2024moderating} on uniform 1-Euclidean preferences extends naturally to path graphs, although with some additional messiness caused by discretizing the interval. Recall that a path is a graph with edges $\{1, 2\}, \{2, 3\}, \dots, \{n-1, n\}$.

\begin{proposition}
  
The minimal IRV exclusion zone of the path on $n$ nodes is $S = \{ \lceil n/6 + 1/2 \rceil, \dots,  n - \lceil n/6 + 1/2 \rceil + 1 \}$.
\end{proposition}
\begin{proof}
First we'll show $S$ is an exclusion zone. Suppose only one candidate $x$ remains in $S$. If all the remaining candidates are on one side of $S$, then $x$ has more than half the vote share and wins. So suppose there are candidates on both sides of $S$. Let $\ell = \lceil n/6 + 1/2 \rceil$ and $r = n - \ell + 1$. The vote share of $x$ is minimized if there are candidates at $\ell-1$ and $r+1$, in which case $x$'s vote share is $1$ (the node $x$ is at) plus half of the remaining vote share between $\ell-1$ and $r+1$:
\begin{align*}
  v(x) &\ge 1 + (r - \ell +1 - 1 ) / 2\\
  &= 1 + \left(n - \lceil n/6 + 1/2 \rceil + 1  - \lceil n/6 + 1/2 \rceil  \right)/2\\
  &=  n/2 - \lceil n/6 + 1/2 \rceil + 3/2\\
  &>  n/2 - \lceil n/6 + 3/2 \rceil + 3/2\\
  &\ge  n/2 -  n/6 - 3/2 + 3/2\\
  &= n/3.
\end{align*}
Thus $x$ has more than a third of the remaining vote share and cannot be eliminated next (as there are at least three candidates remaining), so $S$ is an exclusion zone. 

Next, we'll show $S$ is minimal.  Consider the election with three candidates at $\lceil n/2 \rceil$, $\lceil n/6 + 1/2 \rceil$, and $n - \lceil n/6 + 1/2 \rceil + 1$. The vote share of the candidate at $\lceil n / 2 \rceil$ is
\begin{align*}
  1 + \left(n - \lceil n/6 + 1/2 \rceil - 1 - \lceil n/6 + 1/2 \rceil  - 1\right) /2 &= n/2 - \lceil n/6 + 1/2 \rceil.
\end{align*}
Meanwhile, the vote share of the candidate at $\lceil n/6 + 1/2 \rceil$ is
\begin{align*}
  \lceil n/6 + 1/2 \rceil + (\lceil n/2 \rceil - 1 - \lceil n/6 + 1/2 \rceil) / 2 &= (\lceil n/2 \rceil - 1 + \lceil n/6 + 1/2 \rceil) / 2.
\end{align*}
This is larger, which can be seen since $(\lceil n/2 \rceil - 1 + \lceil n/6 + 1/2 \rceil) / 2 - (n/2 - \lceil n/6 + 1/2 \rceil) = (\lceil n/2 \rceil + 3\lceil n/6 + 1/2 \rceil - n - 1) / 2$ and $\lceil n/2 \rceil + 3\lceil n/6 + 1/2 \rceil \ge n + 3/2$. The vote share of the rightmost candidate is at least as large as the leftmost candidate's, so the center candidate is eliminated. Moreover, the center node is a Condorcet winner (this is the center for odd $n$ and the left center for even $n$). By the Condorcet Chain Lemma, this configuration thus shows that the nodes $\lceil n/6 + 1/2 \rceil$ and $n - \lceil n/6 + 1/2 \rceil + 1 $ are in the minimal exclusion zone, since either can win in this configuration after a tiebreak. All nodes closer to the center must also be in the minimal exclusion zone, since they defeat those two nodes in pairwise contests. 
\end{proof}

Next, we consider a case that demonstrates how exclusion zones can behave in unexpected ways: perfect binary trees. A height-$h$ perfect binary tree can be defined recursively: a height-0 perfect binary tree is a single node and a height-$(h+1)$ perfect binary tree is a node with two neighbors, each of which is the root of a height-$h$ perfect binary tree. We show that perfect binary trees have nontrivial IRV exclusion zones if and only if they have even (and nonzero) height. 

\begin{theorem}\label{thm:binary-tree-odd}
  Perfect binary trees with odd height have no nontrivial IRV exclusion zones.
\end{theorem}
\begin{proof}
    Suppose the height of the tree is $h = 2c + 1$. Such a perfect binary tree consists of a perfect binary tree of height $c$ where we add to each leaf two subtrees which are also perfect binary trees of height $c$. Place one candidate at the root and one candidate in an arbitrary leaf of each of the lower subtrees of height $c$ that we glued together. In such a configuration, the vote share of every candidate is $2^{c+1}-1$: every candidate gets all of the votes from within their perfect binary subtree. Thus, the root can be eliminated next by tiebreak. Since the leaves are Condorcet losers and the root is a Condorcet winner, the minimal IRV exclusion zone is trivial by \Cref{lemma:chain}.
\end{proof}
With even-height trees, any such attempt at a configuration does not cause a tie, and instead favors the root. Instead, we can show using an inductive argument that the set of internal nodes is the minimal IRV exclusion zone of even-height perfect binary trees.

\begin{theorem}
  For perfect binary trees with even height $h > 0$, the minimal IRV exclusion zone is the set of internal nodes.
\end{theorem}
\begin{proof}
Suppose only one candidate at an internal node remains during the running of IRV; call this node $u$. We'll show inductively that $u$ cannot be eliminated before the candidates in the leaves of the subtree rooted at $u$'s $i$th ancestor (by induction on $i$, stepping up a level each time we exhaust candidates in the leaves in the current subtree). 

First, consider the base case, where there is a candidate remaining in a leaf $x$ of the subtree rooted at $u$ itself. 
\begin{enumerate}
  \item Suppose first that $u$ has even height $h = 2c$. Then there is a node $v$ halfway between $u$ and $x$ on the shortest path from $u$ to $x$, where $v$ has height $c$. There are $2^{c}-1$ nodes strictly closer to $x$ than to $u$ (those in the subtree under $v$ that contains $x$). Meanwhile, there are at least $2^{c}-1$ nodes strictly closer to $u$ than to any leaf (those in the subtree rooted at $u$ with height $\ge c+1$). The other nodes in the subtree rooted at $v$ are equidistant from $u$ and $x$, so they contribute equally to their vote shares. Meanwhile, $u$ also gets a contribution from the other height $c$ nodes in its subtree (which it may share with leaves), so $u$ has strictly larger vote share than $x$.
  \item Now suppose that $u$ has odd height $h = 2c+1$. All nodes in this subtree at height $c+1$ and higher, of which there are $2^{c+1} - 1$, are closer to $u$ than to any leaf. Meanwhile, consider the vote share of the leaf $x$. The only nodes closer to $x$ than $u$ are those in the subtree rooted at the height $c$ ancestor of $x$, as all paths to other nodes in the subtree pass through a node of height $c+1$ or higher. Thus the vote share of $x$ is at most $2^{c+1} -1$. Finally, since the tree has even height, $u$ cannot be the root, so $u$ gets additional votes from its ancestors and therefore has strictly larger vote share than $x$.
\end{enumerate}

Now we consider the inductive case, where there are no remaining candidates remaining in the leaves of the subtree rooted at $u$'s $i$th ancestor. We'll show than $u$ cannot be eliminated before a candidate in a leaf in the subtree rooted at $u$'s $(i+1)$th ancestor. This is easy, since if there are no candidates in the leaves of the subtree rooted at $u$'s $i$th ancestor, then every node in this subtree is closer to $u$ than to a leaf in the other subtree of $u$'s ($i+1$)th ancestor. Moreover, $u$'s $i+1$th ancestor also votes for $u$ over any leaf. Thus $u$ has a strictly larger vote share than any candidate in a leaf of the subtree rooted at $u$'s $(i+1)$th ancestor. The candidate at $u$ therefore survives elimination until no leaves remain in this subtree.

By induction, we thus find that $u$ first survives until all leaves below it are eliminated, then survives until all leaves in the subtree of its $i$th ancestor are eliminated, for $i=1, 2,\dots$. Once its $i$th ancestor is the root, $u$ is the last remaining candidate and wins. 

We now show that this exclusion zone is minimal. Consider the configuration from the proof of \Cref{thm:binary-tree-odd}, where the root can be eliminated against a collection of leaves. Taking this configuration and adding another layer of nodes at the bottom of the tree, we are left with a configuration in an even-height perfect binary tree where a collection of nodes at height $1$ defeat the root. Since the root is a Condorcet winner, these nodes at height $1$ must also be in the minimal exclusion zone by \Cref{prop:exclusion-facts}. We could have picked any height-1 nodes by symmetry, so in fact all height-1 nodes are in the minimal exclusion zone. Any of these height-1 nodes is defeated by a higher internal node in a pairwise contest, so every internal node is in the minimal exclusion zone by \Cref{prop:exclusion-facts}. 
\end{proof}

\subsection{IRV exclusion zones in real-world graphs}

\begin{figure}
\centering
\includegraphics[width=0.45\textwidth]{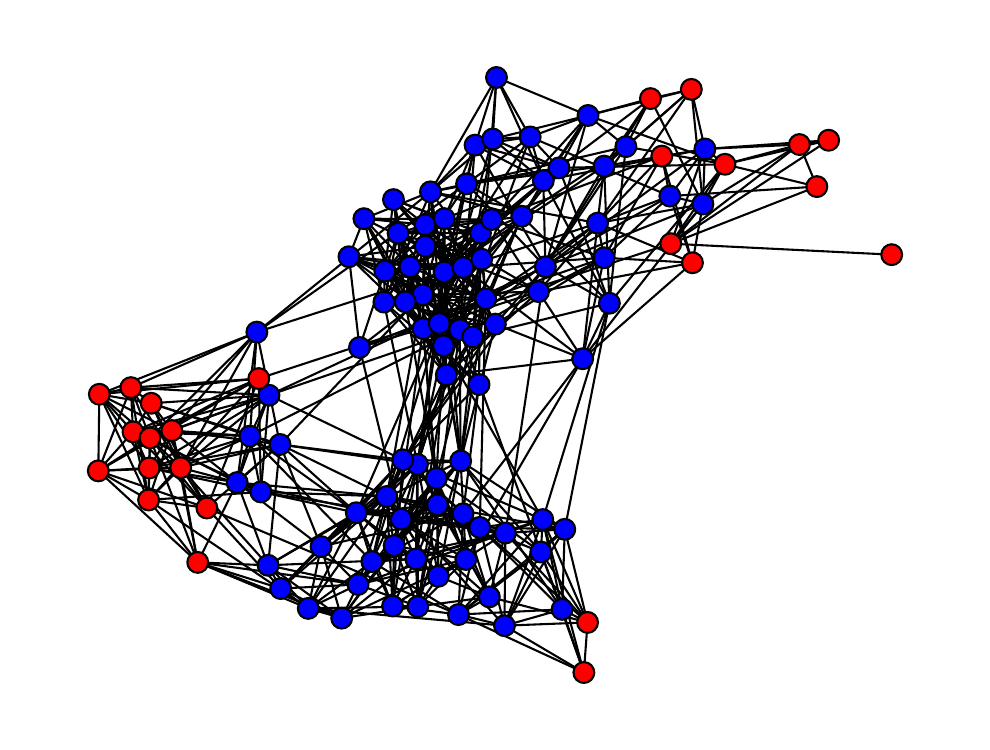}~
\includegraphics[width=0.45\textwidth]{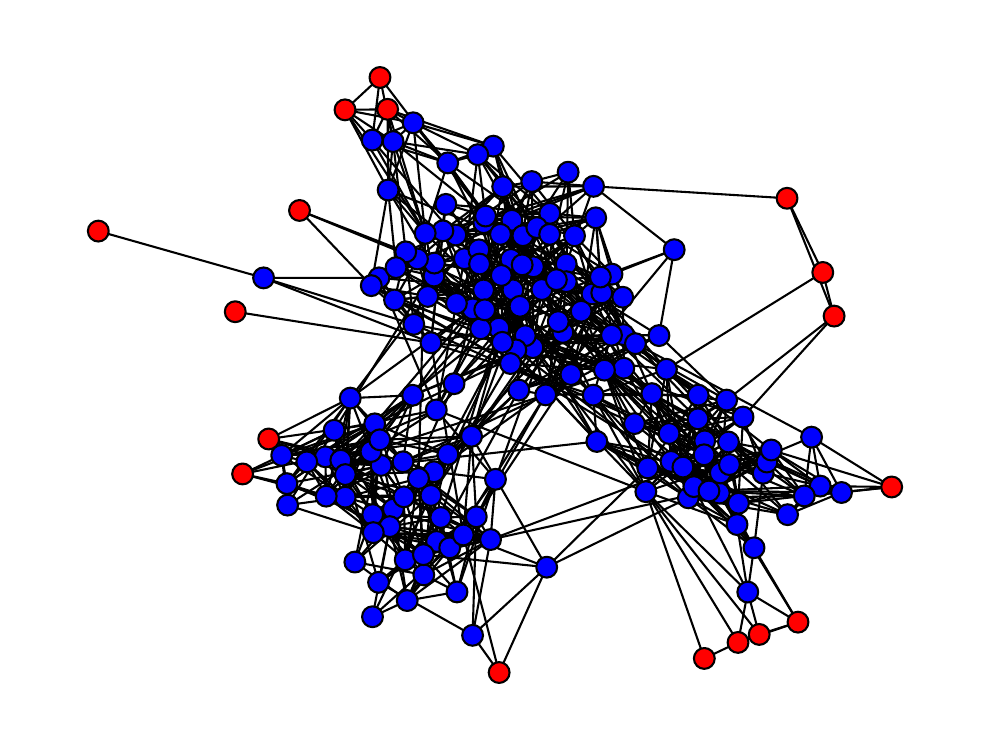}
\caption{Two of the school social networks from \citet{paluck2016data} (ID 5 on the left and 50 on the right) with probable 0.99-approximate IRV exclusion zones in blue and excluded nodes in red. Nodes that cannot win under IRV can include fringe nodes or entire communities.}  \label{fig:schools} 
\end{figure}

 We now ask whether real-world social networks, especially ones where we might expect distance-based preferences, have nontrivial IRV exclusion zones. To this end, we use a collection of 56 social networks from public middle schools in New Jersey~\cite{paluck2016changing,paluck2016data}, where edges represent pairs of students who spend time with each other. If IRV were used for a class president election under graph-based preferences, for instance, any exclusion zone in the graphs would tell us which nodes have a chance of winning and which could not win any election against a member of the exclusion zone. The graphs in this dataset have 110--844 nodes, with an average of 435 nodes, making them much too large for the exact algorithm we applied in the exhaustive search of small graphs. Instead, we apply our approximation algorithm from \Cref{thm:approx} with parameters $\epsilon = \delta = 0.01$ to find node sets that are $0.99$-approximate IRV exclusion zones with probability at least 0.99. We find that 23 of the 56 schools verifiably have no nontrivial IRV exclusion zones, owing to the fact that the algorithm returned the full node set (and any set returned by the algorithm must be a subset of the minimal exclusion zone). The remaining 33 school graphs (59\%) returned nontrivial probable approximate exclusion zones. Even when they are nontrivial, these (probable, approximate) exclusion zones tend to be large, comprising $92\%$ of the graph on average. See \Cref{fig:schools} for visualizations of two of the smaller schools, with the node sets returned by the approximation algorithm in blue. One of the node sets was a notable outlier, consisting of only 9\% of the graph (for school ID 26), whereas all others contained at least 77\% of the graph. We suspected this outcome was a failure of the randomized algorithm, which would be unsurprising given the number of runs (56 schools) and the 1\% failure rate we selected. Indeed, rerunning the algorithm with $\epsilon = \delta = 0.005$ for that single school returned a node set containing $98.5\%$ of the nodes. After this correction, the average size of the nontrivial approximate exclusion zones was $95\%$.
 
 In combination with the exhaustive search of small graphs and trees, these experiments on real-world networks demonstrate that nontrivial IRV exclusion zones frequently exist in practice. However, for real-world networks, they tend to be quite large, revealing only a small fraction of nodes that cannot win under IRV. Our approximation algorithm also allows us to explore IRV exclusion zones in graphs orders of magnitude larger than the exact algorithm.  See \url{https://github.com/tomlinsonk/irv-exclusion-zones} for all of our code and results, including Python and C implementations of our approximation algorithm.

\section{Discussion}

Exclusion zones are a relatively new addition to the landscape of voting theory, and their phenomenology is still not well understood following our initial results on one-dimensional IRV and plurality~\cite{tomlinson2024moderating}.  In this work, we have built up a much more general understanding of exclusion zones, answered some of the questions left open by our prior work, and provided new results about the complexity of computing exclusion zones.

In particular, we have characterized structural properties and resolved an open question about exclusion zones of IRV in higher-dimensional preference spaces, showing that with uniform voters over any $d$-dimensional hyperrectangle with $d >1$, there are no nontrivial IRV exclusion zones. In the graph voting setting, we showed that IRV exclusion zones are computationally difficult to identify, but also provided an efficient approximation algorithm. Through computational experiments, we have also discovered that nontrivial exclusion zones are relatively abundant in small graphs, and likely also present in many larger graphs arising in social network data.  While our focus has been on IRV, the ideas we develop around exclusion zones, such as the nesting of exclusion zones and the Condorcet Chain Lemma, apply to any voting system over any metric space. This suggests that exclusion zones may be a valuable tool in understanding the behavior of other voting systems, a promising line of inquiry for future work. 

The major open questions from our work concern which voting rules have nontrivial exclusion zones in which preference spaces. Perhaps the simplest metric preference space is the 1-Euclidean domain, where we already understand the exclusion zones of IRV, plurality, and Condorcet methods. Do other non-Condorcet methods, like approval, Borda, Coombs, and top-two runoff, have nontrivial exclusion zones with 1-Euclidean preferences? In higher dimensions, are there election sequences satisfying the Condorcet Chain Lemma for any of these other non-Condorcet methods? Whenever there are Condorcet winners, the minimal exclusion zone of any Condorcet method is the set of Condorcet winners, but there may still be something interesting to say about the particular nesting of all exclusion zones with Condorcet methods. For instance, in one dimension, the set of exclusion zones of Condorcet methods captures the essence of the Median Voter Theorem. When there is a Condorcet winning position in higher-dimensional preferences spaces (as with uniform voters over hyperrectangles), is there something we can say about the set of exclusion zones of Condorcet voting rules? 

In terms of improvements in proof techniques, the Condorcet Chain Lemma can only show that the minimal exclusion zone is trivial in preferences spaces where there are both Condorcet winners and losers. In cases with no Condorcet winners, is there another way to certify triviality? Note that there do exist preference spaces with no Condorcet winners and a trivial minimal exclusion zone: in work addressing different questions than exclusion zones, \citet[Figure 2]{skibski2023closeness} shows a 12-node graph with no Condorcet winner, and we have verified using the optimized exhaustive search from \Cref{sec:graph-general} that it has no nontrivial IRV exclusion zone. In the opposite case, when there are nontrivial exclusion zones, our main proof approach with IRV has exploited the structure of the voting system: to argue that $S$ is an exclusion zone, we show that the last candidate in $S$ cannot be eliminated by candidates outside of $S$. This takes advantage of the fact that IRV eliminates candidates one at a time, so if the election began with at least one candidate in $S$, then at some point only one will remain in $S$. However, as we have seen in one dimension, it is possible to argue through other means that voting systems like plurality and Condorcet methods have nontrivial exclusion zones, so we are optimistic that exclusion zones will prove useful beyond IRV in the ongoing endeavor to characterize voting systems.

\section*{Acknowledgments}
This work was supported in part by ARO MURI, a Simons Collaboration grant, a grant from the MacArthur Foundation, a Vannevar Bush Faculty Fellowship, AFOSR grant FA9550-19-1-0183, and NSF CAREER Award \#2143176.
Thanks to Kate Donahue, Jason Gaitonde, Raunak Kumar, Sloan Nietert, Katherine Van Koevering, and the attendees of the AMS Special Session on Mathematics of Decisions, Elections, and Games at the 2025 Joint Mathematics Meetings  for helpful discussions and feedback.

% Bibliography
\bibliographystyle{ACM-Reference-Format}
\bibliography{references}

% Appendix
\appendix
\section{Additional proofs}\label{app:proofs}

\subsection{IRV exclusion zones in higher dimensions}

\begin{proof}[Proof of \Cref{lemma:hyperrect-condorcet}]
First, we show that the center is a Condorcet winner. Consider any opponent at $x = (x_1, \dots, x_d)\ne c$. We need to show they lose to $c$ in a pairwise contest. Let $V$ be the region of voters that are closer to $x$ than $c$. Let $V'$ be the reflection of $V$ across $c$, mapping each point $y = (y_1, \dots, y_d) \in V$ to $y' = (w_1-y_1, \dots, w_d-y_d) \in V'$. Since $d_p(y, y') = 2d_p(c, y')$,
  \begin{align*}
    2d_p(c, y') = d_p(y, y')&\le d_p(x, y') + d_p(x, y) \tag{triangle inequality}\\
    \Rightarrow d_p(x, y') &\ge 2d_p(c, y') - d_p(x, y) \\
        &= d_p(c, y') + d_p(c, y) - d_p(x, y)\tag{since $d_p(c, y) = d_p(c, y')$}\\
        &> d_p(c, y') \tag{since $d_p(c, y) > d_p(x, y)$ by def.\ of $V$}
  \end{align*}
  Thus, $y'$ is closer to $c$ than to $x$. Thus, the vote share of $c$ is at least as large as $x$, since every point in $V'$ votes for $c$. To show that the vote share of $c$ is strictly larger, consider the $L_p$ ball around $c$ with radius $0 < \epsilon < \min\{\min_i\{w_i/2\}, d(c, x)/2\}$ (this is positive since $x$ is not at the center). All of these points in the ball are closer to $c$ than $x$, and none of them are in $V$ or $V'$ (since the ball is its own reflection about $c$). Thus $c$'s vote share is strictly larger than $x's$.
  
  Next, we show that the corners are Condorcet losers. WLOG, suppose the corner is the origin. Consider any opponent $x$ and let $V$ be the region of points closer to $o$ than to $x$. Let $V'$ be the reflection of $V$ across the midpoint of the hyperrectangle. For any point $y = (y_1, \dots, y_d) \in V$ (weakly) closer to the origin than to $x$, let $y' = (w_1 - y_1, \dots, w_d - y_d) \in V'$ be its reflection. We will show that $y'$ is (weakly) closer to $x$ than to the origin.

For any $i \in \{1, \dots, d\}$, consider the scalar function 
\begin{align*}
  f_i(x_i) = |x_i - y_i|^p + |w_i- x_i - y_i|^p 
\end{align*}
Both $|x_i - y_i|$ and $|w_i- x_i - y_i|$ are convex in $x_i$. Since $g(x_i) = x_i^p$ is convex and non-decreasing for $x_i \ge 0$, we then have that $f_i(x_i)$ is convex by convex composition rules. Moreover, $f_i(0) = f_i(w_i) = |y_i|^p + |w_i-y_i|^p$. Thus, $f_i$ achieves its maximum over $[0, w_i]$ at $x_i = 0$, by convexity. We therefore have:
\begin{align*}
&|x_i - y_i|^p + |w_i- x_i - y_i|^p  \le |y_i|^p + |1-y_i|^p \tag{$\forall i = 1, \dots, d$}\\
\Rightarrow \quad & \sum_{i = 1}^d |x_i - y_i|^p + \sum_{i = 1}^d|w_i- y_i - x_i |^p \le \sum_{i = 1}^d |y_i|^p + \sum_{i = 1}^d|1-y_i|^p \\
\Leftrightarrow \quad &  \|x - y\|_p^p + \|y'- x\|_p^p \le \|y\|_p^p + \|y'\|_p^p \\
   \Rightarrow \quad &  \|x - y\|_p^p + \|y'- x\|_p^p \le \|x - y\|_p^p + \|y'\|_p^p \tag{since $\|y\|_p \le \|x - y\|_p$}\\
   \Rightarrow \quad & \|y' - x\|_p^p \le \|y'\|_p^p. 
\end{align*}

That is, $y'$ is (weakly) closer to $x$ than to the origin. So, every point in $V'$ is (weakly) closer to $x$ than to $o$ (and the only point that could be shared by $V$ and $V'$ is the center). At worst, their vote shares are tied and the origin can lose. 
\end{proof}

To prove that the square has a no nontrivial IRV exclusion zone (\Cref{prop:l2-square}), we will need several lemmas establishing each step in the construction. In the first step, we find a configuration defeating the Condorcet winner (the center), which works for any $L_p$ distance.

\begin{lemma}\label{lemma:square-step1}
  With uniform $L_p$ voters over the unit square, the center $c = (1/2, 1/2)$ has the smallest vote share when there are four additional candidates $\ell = (x, 1/2), t = (1/2, 1-x), r = (1-x, 1/2),$ and $b = (1/2, x)$ for $1/2 - 1/\sqrt{5} < x < 1/2$ and any $p \ge 1$.
\end{lemma}

\begin{proof}
For any $p \ge 1$, the points equidistant from the center and each of the four surrounding candidates form vertical and horizontal lines, with the center getting votes from all voters on $c$'s side of these lines. That is, from $c$'s perspective, the configuration looks as follows for any $p \ge 1$:
  \begin{center}
    \begin{tikzpicture}[x=6mm,y=6mm]
%     \fill[color=black!10] (4.5, 4.5) rectangle (5, 5);
%     \fill[color=black!10] (0, 0) rectangle (.5, .5);
%     \fill[color=black!10] (4.5, 0) rectangle (5, .5);
%     \fill[color=black!10] (0, 4.5) rectangle (.5, 5);

  \draw[draw=black, thick] (0, 0) rectangle ++(5, 5);
    
  \draw[dashed] (1.5, 3.5) -- (1.5, 1.5) -- (3.5, 1.5) -- (3.5, 3.5) -- (1.5, 3.5);
%   \draw[dashed] (1.5, 1.5) -- (0.5, 0.5);
%   \draw[dashed] (3.5, 3.5) -- (4.5, 4.5);
%   \draw[dashed] (1.5, 3.5) -- (0.5, 4.5);
%   \draw[dashed] (3.5, 1.5) -- (4.5, 0.5);

  \node[fill=red, circle,inner sep=2pt, label=$c$]  at (2.5, 2.5) {};
  \node[fill, circle,inner sep=2pt, label=0:$b$] (M) at (2.5, 0.5) {};
    \node[fill, circle,inner sep=2pt, label=0:$t$]  at (2.5, 4.5) {};
    \node[fill, circle,inner sep=2pt, label=$\ell$]  at (0.5, 2.5) {};
    \node[fill, circle,inner sep=2pt, label=$r$]  at (4.5, 2.5) {};
\end{tikzpicture}
  \end{center}
  
  The center candidate's vote share is $(1/2(1-2x))^2 = x^2 - x + 1/4$. Since the other candidate evenly divide the remaining vote share by symmetry, $c$ is first eliminated if $    x^2 - x + 1/4 < 1/5$. We can use the quadratic formula to find this occurs for $ 1/2 -  1/\sqrt{5} < x < 1/2 + 1/\sqrt{5}$. 
\end{proof}

Next, we find a configuration where the candidate $b$ is eliminated first.

\begin{lemma}\label{lemma:square-step2}
      With uniform $L_2$ votes over the unit square, a candidate at $b = (1/2, x)$ has the smallest vote share against candidates at $\ell = (0, 1/2)$ and $r = (1, 1/2)$ for $0 \le x \le 0.16$. 
\end{lemma}
  \begin{proof}
This configuration looks as follows:

  \begin{center}
   \begin{tikzpicture}[x=10mm,y=10mm]
  \draw[draw=black, thick] (0, 0) rectangle ++(5, 5);

  \draw[dotted] (0, 2.5) -- (2.5, 0.5);
  \draw[dotted] (5, 2.5) -- (2.5, 0.5);
    \draw[dotted] (2.5, 0.5) -- (5, 0.5);
  \draw[dotted] (0, 3.0625) -- (2.5, 3.0625);
  
\draw[dashed] (0.05, 0) -- (2.5, 3.0625);

\draw[dashed] (4.95, 0) -- (2.5, 3.0625);
  \draw[dashed] (2.5, 5) -- (2.5, 3.0625);
  
  \node[fill, circle,inner sep=2pt, label=45:$\ell$] (L) at (0, 2.5) {};
  \node[fill=red, circle,inner sep=2pt, label=270:$b$] (M) at (2.5, 0.5) {};
  \node[fill, circle,inner sep=2pt, label=0:$r$] (R) at (5, 2.5) {};

    \draw [decorate,decoration={
        brace,
        mirror,
        raise=1.5mm
    }] (5, 0) -- (5, 0.5) node [pos=0.5,anchor=west,xshift=3mm] {$x$}; 
    
    \draw [decorate,decoration={
        brace,
        raise=2mm
    }] (0, 0) -- (0, 3.0625) node [pos=0.5,anchor=east,xshift=-2.3mm] {$\frac{1 - 2 x^2}{2 - 4 x}$}; 

      \draw [decorate,decoration={
        brace,
        mirror,
        raise=2mm
    }] (0.07, 0) -- (2.5, 0) node [pos=0.5,anchor=north,yshift=-3mm] {$1/2 - x^2$}; 
\end{tikzpicture}
\end{center}
The midpoint of the line $b\ell$ is at $(1/4, 1/2 - (1/2 - x)/2) = (1/4, 1/4 + x/2)$ and the slope of $b \ell$ is $ -(1/2 - x) / (1/2) = 2x - 1$, so the slope of its perpendicular is $1/(1-2x )$. Thus, the height of the triangle of $b$'s voter region is $1/(1-2x)(1/2 - 1/4) + 1/4 + x/2 = \frac{1 - 2 x^2}{2 - 4 x}$ (plugging in $1/2$ into the point-slope form for the $\ell b$ perpendicular). Note that this triangular vote share configuration only occurs while $\frac{1 - 2 x^2}{2 - 4 x} < 1$. Solving for $x$ using the quadratic formula shows that this is satisfied while $x < 1 - 1/\sqrt{2}\approx 0.29$.

 Meanwhile, the $\ell b$ perpendicular intersects the $x$ axis at $z$ such that $1/(1-2x)(z - 1/4) + 1/4 + x/2 = 0$. Solving for $z$ yields $z = x^2$. Thus, $b$'s vote share is 
\begin{align*}
  \frac{1 - 2 x^2}{2 - 4 x}(1/2 - x^2) &= \frac{(1 - 2 x^2)^2}{4 - 8 x}
\end{align*}
This is the smallest if 
\begin{align*}
  &\frac{(1 - 2 x^2)^2}{4 - 8 x}< 1/3 \\
  \Leftrightarrow \quad &  ( 4 - 8 x)/3 - (1 - 2 x^2)^2 > 0\\
  \Leftrightarrow \quad & 1/3 - 8 x/3 + 4 x^2 - 4 x^4 > 0.
\end{align*}
In principle, we could solve this exactly with the (very messy) quartic formula; Mathematica gives that this is satisfied for $0 \le x \le 0.16$ (the exact bound is closer to 0.1645). 
\end{proof}

We can now eliminate $\ell$.

\begin{lemma}\label{lemma:square-step3}
    With uniform $L_p$ votes over the unit square, a candidate at $\ell = (0, 1/2)$ has the smallest vote share against candidates at $b = (0, x)$ and $t = (0, 1-x)$ for $1/6 < x < 1/2$ and any $p\ge 1$. 
\end{lemma}
\begin{proof}
This configuration looks as follows, for any $p \ge 1$:
    \begin{center}
   \begin{tikzpicture}[x=10mm,y=10mm]
  \draw[draw=black, thick] (0, 0) rectangle ++(5, 5);

  \draw[dashed] (0, 1.7)  -- (5, 1.7); 

  \draw[dashed] (0, 3.3)  -- (5, 3.3);

  \node[fill, circle,inner sep=2pt, label=0:$b$] (L) at (0, 0.9) {};
  \node[fill=red, circle,inner sep=2pt, label=0:$\ell$] (M) at (0, 2.5) {};
  \node[fill, circle,inner sep=2pt, label=0:$t$] (R) at (0, 4.1) {};
  
      \draw [decorate,decoration={
        brace,
        raise=2mm
    }] (0, 4.1) -- (0, 5) node [pos=0.5,anchor=east,xshift=-2.3mm] {$x$}; 
        \draw [decorate,decoration={
        brace,
        raise=2mm
    }] (0, 0) -- (0, 0.9) node [pos=0.5,anchor=east,xshift=-2.3mm] {$x$}; 
\end{tikzpicture}
\end{center}
Candidate $\ell$ gets vote share $(1-2x)/2 = 1/2 - x$, which is the smallest if $1/2 - x < 1/3 \Leftrightarrow x > 1/6.$

\end{proof}

Finally, we find a configuration where a corner wins against $t$.

\begin{lemma}\label{lemma:square-step4}
  With uniform $L_2$ votes over the unit square, a candidate at $t = (0, 1-x)$ has the smallest vote share against candidates at $\ell = (0, 1)$ and $r = (1, 1)$ for $0 \le x \le 0.19$. 
\end{lemma}

\begin{proof}
This configuration looks as follows:

  \begin{center}
   \begin{tikzpicture}[x=10mm,y=10mm]
  \draw[draw=black, thick] (0, 0) rectangle ++(5, 5);

  \draw[dashed] (0, 2) -- (3, 2);
  \draw[dashed] (5, 0) -- (3, 2);

  \draw[dotted] (0, 4) -- (5, 5);
  
  \draw[dashed] (2.4, 5)  -- (3, 2); 
  \draw[dotted] (3, 2) -- (3, 5) ;

  \node[fill, circle,inner sep=2pt, label=180:$\ell$] (L) at (0, 0) {};
  \node[fill=red, circle,inner sep=2pt, label=45:$t$] (M) at (0, 4) {};
  \node[fill, circle,inner sep=2pt, label=0:$r$] (R) at (5, 5) {};
  
    \draw [decorate,decoration={
        brace,
        raise=2mm
    }] (0, 4.01) -- (0, 5) node [pos=0.5,anchor=east,xshift=-2.3mm] {$x$}; 
    
    \draw [decorate,decoration={
        brace,
        raise=2mm
    }] (0, 2) -- (0, 3.99) node [pos=0.5,anchor=east,xshift=-2.3mm] {$1/2 - x/2$}; 
              \draw [decorate,decoration={
        brace,
        mirror,
        raise=2mm
    }] (0.04, 2) -- (3, 2) node [pos=0.5,anchor=north,yshift=-3mm] {$1/2 + x/2$}; 
    
      \draw [decorate,decoration={
        brace,
        raise=2mm
    }] (2.4, 5) -- (3, 5) node [pos=0.5,anchor=south,yshift=2mm] {$x(1/2 + x/2)$}; 
\end{tikzpicture}
\end{center}
Candidate $t$'s vote share is
\begin{align*}
  v(m) &= (1/2 + x/2)^2 - x(1/2 + x/2)^2/2\\
  &= 1/4 + 3 x/8 - x^3/8.
\end{align*}
Candidate $\ell$ has smaller vote share than $r$, at
\begin{align*}
 v(\ell) &= (1/2 - x/2)(1/2 + x/2) + (1/2 - x/2)^2/2\\
 &= 3/8 - x/4 - x^2/8
\end{align*}
Thus, candidate $t$ has the smallest vote share if
\begin{align*}
  &1/4 + 3 x/8 - x^3/8 < 3/8 - x/4 - x^2/8\\
  \Leftrightarrow \quad & x^3/8 - x^2 /8 - 5/8x +1/8 > 0.
\end{align*}
Note that the derivative of this function, $3x^2/8 - x/4 - 5/8$, is negative for $x \in [0, 1]$. Moreover, $x^3/8 - x^2 /8 - 5/8x +1/8$ is positive for $x = 0.19$ (taking value $\approx 0.003$), so $t$ has the smallest vote share for all $x < 0.19$ (the exact root of the cubic is approximately 0.194, but doesn't admit a closed form in terms of real radicals).
\end{proof}

We can then chain together the above constructions.
\begin{proof}[Proof of \Cref{prop:l2-square}]
By \Cref{lemma:hyperrect-condorcet}, the center of the square is a Condorcet winner and the corners are Condorcet losers. \Cref{lemma:square-step1,lemma:square-step2,lemma:square-step3,lemma:square-step4} then give us a sequence of configurations satisfying \Cref{lemma:chain} (once the previous winner is eliminated, the next desired winner can win a tie-breaker). Thus the square has no nontrivial exclusion zones with uniform $L_2$ voters by \Cref{lemma:chain}. 
\end{proof}

\begin{proof}[Proof of \Cref{lemma:rect-to-hyper}]
Let $R = [0, w_1]\times\dots \times [0, w_d]$ be a hyperrectangle with dimension $d\ge 2$ and suppose that for any rectangle with uniform $L_p$ voters, there exists a sequence of configurations satisfying \Cref{lemma:chain}. 

  Suppose all candidates share all but two of their coordinates in $R$; i.e., they lie on a single plane $P$ parallel to a pair of axes. WLOG suppose the non-shared coordinates are the first and second. Let $a$ and $b$ be two candidates and consider any voter at a point $x$. Let $x'$ be the projection of $x$ on $P$. If $x'$ is closer to $a$ than to $b$, then
  \begin{align*}
    \|x' - a\|_p^p > \|x' - b\|_p^p\\
    \Leftrightarrow \quad & \sum_{i = 1}^d |x'_i - a_i|^p> \sum_{i = 1}^d |x'_i - b_i|^p \\
    \Leftrightarrow \quad &  \sum_{i = 1}^2 |x'_i - a_i|^p> \sum_{i = 1}^2 |x'_i - b_i|^p\tag{$x', a, b$ are all in $P$}\\
    \Leftrightarrow \quad &  \sum_{i = 1}^2 |x'_i - a_i|^p + \sum_{i = 3}^d |x_i - a_i|^p> \sum_{i = 1}^2 |x'_i - b_i|^p +  \sum_{i = 3}^d |x_i- a_i|^p\\
\Leftrightarrow \quad &  \sum_{i = 1}^d |x_i - a_i|^p > \sum_{i = 1}^d |x_i - b_i|^p \\
\Leftrightarrow \quad &  \|x - a\|_p^p > \|x - b\|_p^p.
  \end{align*}
  Thus, $x$ is also closer to $a$ than to $b$. That is, all points in $R$ vote in the same way as their projections onto $P$. Moreover, the voter distribution remains uniform after this projection, since $P$ is perpendicular to a pair of axes and $R$ is a hyperrectangle. Thus, vote shares are the same as if the election were in the rectangle formed by projecting $R$ onto $P$. 
  
  This allows us to apply a sequence of configurations in a rectangle within $R$. By \Cref{lemma:hyperrect-condorcet}, the center $c = (w_1 / 2, \dots, w_d/2)$ of $R$ is a Condorcet winner. Thus the minimal exclusion zone of $R$ contains $c$, by \Cref{prop:exclusion-facts}. Consider the plane $P = (x, y, w_3/2, \dots, w_d/2)$. The intersection $R \cap P$ forms a rectangle, in which we can apply a sequence of configurations satisfying \Cref{lemma:chain}. By the above reasoning, these elections play out the same way in $R$ as they do in the rectangle. Thus, repeated applications of \Cref{prop:exclusion-facts} show that every point in $R \cap P$ is in the minimal exclusion zone of $R$, including the point $(0, w_2/2, \dots, w_d/2)$. We can then repeat with the new plane $P' = (0, x, y, w_4/2, \dots, w_d / 2)$, with one more coordinate zeroed out, finding that $(0, 0, w_3/2, \dots, w_d/2)$ is in the minimal exclusion zone of $R$. Repeat the procedure $d$ times in total, and we find that the origin $(0, \dots, 0)$ is in the minimal exclusion zone of $R$. Since the origin is a Condorcet loser by \Cref{lemma:hyperrect-condorcet}, we then have a sequence of configurations satisfying \Cref{lemma:chain} in $R$, showing its minimal exclusion zone is trivial. 
\end{proof}

\begin{lemma}\label{lemma:rectangle-l1-corner-squeeze}
  Suppose we have uniform $L_1$ voters over the rectangle $[0, w] \times [0, 1]$, $w > 1$. With candidates at $m = (0, 1- c - \epsilon)$, $\ell = (0, c)$ , and $r = (w-c, 1)$ for $0 \le c < 1/2$ and $0 \le \epsilon < \min\{1/2 - c, \frac{1}{8w}\} $, $r$ is a possible IRV winner.
\end{lemma}
\begin{proof}
This configuration divides up the rectangle as follows:

\begin{center}
   \begin{tikzpicture}[x=8mm,y=8mm]
  \draw[draw=black, thick] (0, 0) rectangle ++(10, 6);

  \draw[dashed] (5, 3) -- (7, 1) -- (7, 0);
  \draw[dashed] (0, 3) -- (5, 3);
  \draw[dashed] (5, 5) -- (4, 6);
  \draw[dashed] (5, 3) -- (5, 5);
    
  \draw[dotted] (9, 6) -- (9, 5);
    \draw[dotted] (0, 5) -- (9, 5);

  \draw[dotted] (5, 0) -- (5, 3);
  \draw[dotted] (0, 1) -- (7, 1);

  \node[fill, circle,inner sep=2pt, label=45:$\ell$] (L) at (0, 1) {};
  \node[fill=red, circle,inner sep=2pt, label=45:$m$] (M) at (0, 5) {};
  \node[fill, circle,inner sep=2pt, label=135:$r$] (R) at (9, 6) {};
  
    \draw [decorate,decoration={
        brace,
        raise=2mm
    }] (0, 5.01) -- (0, 6) node [pos=0.5,anchor=east,xshift=-2.3mm] {$c + \epsilon$}; 
    
        \draw [decorate,decoration={
        brace,
        raise=2mm
    }] (0, 0) -- (0, 1) node [pos=0.5,anchor=east,xshift=-2.3mm] {$c$};

      \draw [decorate,decoration={
        brace,
        mirror,
        raise=2mm
    }] (0.04, 3) -- (5, 3) node [pos=0.5,anchor=north,yshift=-3mm] {$(w + \epsilon)/2$}; 
    
        \draw [decorate,decoration={
        brace,
        raise=2mm
    }] (9, 6) -- (10, 6) node [pos=0.5,anchor=south,yshift=3mm] {$c$};
    
      \draw [decorate,decoration={
        brace,
        raise=2mm
    }] (0, 3) -- (0, 4.99) node [pos=0.5,anchor=east,xshift=-2.3mm] {$(1 - 2c - \epsilon)/2$};

        \draw [decorate,decoration={
        brace,
        mirror,
        raise=2mm
    }] (5, 0) -- (7, 0) node [pos=0.5,anchor=north,yshift=-3mm] {$(1 - 2c - \epsilon)/2$}; 
\end{tikzpicture}
\end{center}

The vote shares are:
\begin{align*}
  v(m)   &= (c +\epsilon)(w-c)/2 + (w+\epsilon)(1-2c-\epsilon)/4\\
  &= -c^2/2 + \epsilon/4 - c \epsilon - \epsilon^2/4 + w/4 + \epsilon w/4 \\[0.5em]
  v(\ell)   &= (w+\epsilon)(1-2c-\epsilon)/4 + c(1 + w -2c)/2 + (1-2c - \epsilon)^2/8\\
  &= 1/8 - c^2/2 - \epsilon^2/8 + w/4 - \epsilon w/4 \\[0.5em]
    v(r)   &= (w - \epsilon)/2 + (c+\epsilon)^2 / 2 - (1-2c - \epsilon)^2/8 - c(1-2c-\epsilon)/2\\
    &= -1/8 + c^2 - \epsilon/4 + c \epsilon + 3 \epsilon^2/8 + w/2
\end{align*}

Suppose $0 \le \epsilon < \frac{1}{8w}$. Then,
\begin{align*}
  v(\ell) - v(m) &= 1/8 - c^2/2 - \epsilon^2/8 + w/4 - \epsilon w/4 - (-c^2/2 + \epsilon/4 - c \epsilon - \epsilon^2/4 + w/4 + \epsilon w/4)\\
  &= 1/8 - \epsilon/4 + c \epsilon + \epsilon^2/8 - \epsilon w/2\\
  &> \frac{1}{8} - \frac{1}{32w} - \frac{1}{16}\\
  &> \frac{1}{8} - \frac{1}{32} - \frac{1}{16} = 1/32.
  \end{align*}
  Thus $m$ has a smaller vote share than $\ell$. Similarly, $m$ has a smaller vote share than $r$:
   \begin{align*}
  v(r) - v(m) &= -1/8 + c^2 - \epsilon/4 + c \epsilon + 3 \epsilon^2/8 + w/2 - (-c^2/2 + \epsilon/4 - c \epsilon - \epsilon^2/4 + w/4 + \epsilon w/4)\\
  &= -1/8 + 3 c^2/2 - \epsilon/2 + 2 c \epsilon + 5 \epsilon^2/8 + w/4 - \epsilon w/4 \\
  &> -1/8 - \frac{1}{16w} + w/4 - 1/32\\
  &> - 1/8 - 1/16 + 1/4 - 1/32 = 1/32.
  \end{align*}
  Thus $m$ is eliminated first. After $m$'s elimination, the new configuration looks as follows:
  
\begin{center}
   \begin{tikzpicture}[x=8mm,y=8mm]
  \draw[draw=black, thick] (0, 0) rectangle ++(10, 6);

  \draw[dashed] (2, 6) -- (7, 1) -- (7, 0);
%  \draw[dashed] (0, 3) -- (5, 3);
%  \draw[dashed] (5, 5) -- (4, 6);
%  \draw[dashed] (5, 3) -- (5, 5);
    
%  \draw[dotted] (9, 6) -- (9, 5);
%
%  \draw[dotted] (5, 0) -- (5, 3);
  \draw[dotted] (0, 1) -- (9, 1) -- (9, 6);

  \node[fill, circle,inner sep=2pt, label=45:$\ell$] (L) at (0, 1) {};
  \node[fill, circle,inner sep=2pt, label=135:$r$] (R) at (9, 6) {};

        \draw [decorate,decoration={
        brace,
        raise=2mm
    }] (0, 0) -- (0, 1) node [pos=0.5,anchor=east,xshift=-2.3mm] {$c$}; 

        \draw [decorate,decoration={
        brace,
        raise=2mm
    }] (9, 6) -- (10, 6) node [pos=0.5,anchor=south,yshift=3mm] {$c$};
    
        \draw [decorate,decoration={
        brace,
        mirror,
        raise=2mm
    }] (0, 0) -- (7, 0) node [pos=0.5,anchor=north,yshift=-3mm] {$(w + 1 - 2c)/2$}; 
    
\end{tikzpicture}
\end{center}

The vote shares are now:
\begin{align*}
  v(\ell)& = (1-c)(w - c)/2 + c(w+1-2c)/2\\
  &=  w/2-c^2/2\\[0.5em]
  v(r) &= (1-c)(w - c)/2 + c(1-c) + c(w - 1 + 2c)/2\\
  &= w/2 + c^2/2.
\end{align*}
Thus $r$ can win under IRV (with a tiebreak if $c=0$ and without when $c> 0$).
\end{proof}

\begin{lemma}\label{lemma:rect-minimal-l1}
  In any rectangle $R = [0, w] \times [0, h]$ with uniform $L_1$ voters, there is a sequence of candidate configurations satisfying \Cref{lemma:chain}.
\end{lemma}
\begin{proof}
WLOG suppose $w \ge h$. We can rescale the rectangle so that $h = 1$ without affecting elections, so suppose WLOG that $h = 1$. We know the center of $R$ is a Condorcet winner by \Cref{lemma:hyperrect-condorcet}. Consider the 3-candidate configuration with candidates at $c = (x, 1/2)$, $t = (x - \epsilon, 1/2 + 2\epsilon)$, and  $b = (x - \epsilon, 1/2 - 2\epsilon)$. The following visualization is zoomed in near the center of $R$, showing the boundaries of each candidate's voter share with $L_1$ voters:
 \begin{center}
 \begin{tikzpicture}[x=10mm,y=10mm]

  \draw[dotted] (0, 0) -- (0, 2) -- (1, 2) -- (1, 0) -- (1, -2) -- (0, -2) -- (0, 0) -- (1, 0) ;

  \draw[dashed] (3, 1.5) --  (1, 1.5) -- (0, 0.5) -- (-2, 0.5);

  \draw[dashed] (3, -1.5) --  (1, -1.5) -- (0, -0.5) -- (-2, -0.5);

  \node[fill, circle,inner sep=2pt,label=$t$] (B) at (0, 2) {};
  \node[fill=red, circle,inner sep=2pt, label=0:$c$] (C) at (1, 0) {};
  \node[fill, circle,inner sep=2pt, label=$b$] (T) at (0, -2) {};
\end{tikzpicture}
  \end{center}
  
  Note that $c$'s vote share is strictly less than $4 \epsilon h$, while $t$ and $b$'s vote share is larger than $(w - 4\epsilon)/2$. Pick $\epsilon = \min\{w / 12, 1 / 4\}$ (ensuring the candidates are all in the rectangle). Note that if $\epsilon = 1 / 4$, then $1 / 4 \le w / 12$, so $ w \ge 3$.   We find:
  
  \begin{align*}
    v(c) &< 4\epsilon \\
    &= \begin{cases}
      w /3, &\text{if $\epsilon = w / 12$}\\
      1, &\text{if $\epsilon = h / 4$}
    \end{cases}
  \end{align*}
  and
  \begin{align*}
    v(t) = v(b) &> (w - 4\epsilon)/2 \\
      &= \begin{cases}
      (w - w/3)/2, &\text{if $\epsilon = w / 12$}\\
      (w - 1) /2, &\text{if $\epsilon = 1 / 4$}
    \end{cases}\\
    &\ge \begin{cases}
      w/3, &\text{if $\epsilon = w / 12$}\\
      1, &\text{if $\epsilon = 1 / 4$} \tag{since $w \ge 3$ if $\epsilon = 1/4$}
    \end{cases}\\
    & = v(c).
  \end{align*}
  Thus $c$ is eliminated first. Then, in the two candidate configuration with either $t$ or $b$ and $c' = (x- \epsilon, 1 / 2)$, $c'$ wins with more than half of the vote. We can therefore make a sequence of configurations $C_1, \dots, C_{n}$ where $C_i$ is the above 3-candidate configuration with $x = 1/2 - \epsilon (i - 1) / 2 $ for odd $i$ and the 2-candidate configuration with $x = 1/2 - \epsilon i/2$ for even $i$ where the winner in $C_i$ is eliminated first in $C_{i+1}$ and the winner in $C_1$ is the Condorcet winner. At some even $n$, we will have $x = 0$ (which $n$ depends on whether $\epsilon = w / 12$ or $1/4$; in the first case, it's at $n = 12$, while in the second case, it's at $n \approx 4 w $---we may overshoot $x = 0$ in the second-to-last configuration, but we can  reduce the change in $x$ and preserve the construction). In this way, we get a sequence of chained configurations walking from the center to the left side of the rectangle.
  
  Then, let $C_{n+1}$ be the configuration with candidates at $(0, 1/2)$, $(0, 1/6)$, $(0, 5/6)$. The vote shares are tied at $w/3$, so the middle candidate can be first eliminated.  By symmetry, the candidate at $(0, 5/6)$ can win.
  
  Next, let $C_{n + 2}$ be the configuration from \Cref{lemma:rectangle-l1-corner-squeeze} with $\epsilon = \frac{1}{10w}$ and $c = 1/6 - \frac{1}{10w}$, which includes a candidate at $(0, 5/6)$, but where the candidate at $(w - c, 1)$ is a possible IRV winner. To get back to the left half of the rectangle let $C_{n+3}$ be the configuration with two candidates at $(c, 1)$ and $(w-c, 1)$. By symmetry, the candidate at $(c, 1)$ can win. Then, to get back to the left side of the rectangle, let $C_{n + 4}$ be the configuration with two candidates at $t=(c, 1)$ and $\ell = (0, 5/6 + \frac{1}{11w})$. Every point in the rectangle below $y = 5/6+\frac{1}{11w}$ is closer to $\ell$ than $t$ (in $L_1$ distance), so $\ell$ wins. 
  
  In configurations $C_{n+2}, C_{n+3}, C_{n+4}$, we have moved the winner from $(0, 5/6)$ to $(0, 5/6 + \frac{1}{11w})$. We can repeat this triplet of configurations many times, but with $c = 1/6 - \frac{1}{10w} - \frac{i}{11w}$, keeping $\epsilon = \frac{1}{10w}$ (in the last triplet, we may need a smaller $\epsilon$ so that $c \ge 0$, and we can stop once we have a winner at $(w, 1)$; no need to perform the final reflection). In $\approx 2 w$ such triplets, we reach a winner at $(w, 1)$. 
  
  Since the corners are Condorcet losers by \Cref{lemma:hyperrect-condorcet}, the minimal IRV exclusion zone is trivial by \Cref{lemma:chain}, with a chain of $O(w)$ configurations from the center to a corner.
  
\end{proof}

We now provide the equivalent constructions for uniform $L_2$ voters in a rectangle.

\begin{lemma}\label{lemma:rectangle-l2-corner-squeeze}
  Suppose we have uniform $L_2$ voters over the rectangle $[0, w] \times [0, 1]$, $w > 1$. With candidates at $m = (0, 1- c - \epsilon)$, $\ell = (0, c)$ , and $r = (w, 1-c)$ for $0 \le c < 1/2$ and $0 \le \epsilon < \frac{1 - 2c}{4w^2 + 2} $, $r$ is a possible IRV winner.
\end{lemma}
\begin{proof}
This configuration divides up the rectangle as follows:

\begin{center}
   \begin{tikzpicture}[x=8mm,y=8mm]
  \draw[draw=black, thick] (0, 0) rectangle ++(10, 6);
  \draw[dashed] (4.6, 6) -- (4.778, 3);
  \draw[dashed] (0, 3) -- (4.778, 3);
  \draw[dashed] (6, 0) -- (4.778, 3);
%  \draw[dashed] (5, 3) -- (5, 5);
%    
  \draw[dotted] (4.778, 0) -- (4.778, 6);
%    \draw[dotted] (0, 5) -- (9, 5);
%
%  \draw[dotted] (5, 0) -- (5, 3);
%  \draw[dotted] (0, 1) -- (7, 1);

  \node[fill, circle,inner sep=2pt, label=45:$\ell$] (L) at (0, 1) {};
  \node[fill=red, circle,inner sep=2pt, label=45:$m$] (M) at (0, 5) {};
  \node[fill, circle,inner sep=2pt, label=135:$r$] (R) at (10, 5) {};
  
    \draw [decorate,decoration={
        brace,
        raise=2mm
    }] (0, 5.01) -- (0, 6) node [pos=0.5,anchor=east,xshift=-2.3mm] {$c + \epsilon$}; 
    
        \draw [decorate,decoration={
        brace,
        raise=2mm
    }] (0, 0) -- (0, 1) node [pos=0.5,anchor=east,xshift=-2.3mm] {$c$}; 
%    
%    
%    
%      \draw [decorate,decoration={
%        brace,
%        mirror,
%        raise=2mm
%    }] (0.04, 3) -- (5, 3) node [pos=0.5,anchor=north,yshift=-3mm] {$(w + \epsilon)/2$}; 
%    
        \draw [decorate,decoration={
        brace,
        mirror,
        raise=2mm
    }] (10, 5) -- (10, 6) node [pos=0.5,anchor=west,xshift=3mm] {$c$};
    
      \draw [decorate,decoration={
        brace,
        raise=2mm
    }] (4.778, 0.05) -- (4.778, 2.95) node [pos=0.5,anchor=east,xshift=-2.3mm] {$(1-\epsilon)/2$};  
    
      \draw [decorate,decoration={
        brace,
        mirror,
        raise=2mm
    }] (4.778, 3) -- (4.778, 5.95) node [pos=0.5,anchor=west,xshift=2.3mm] {$(1+\epsilon)/2$};

        \draw [decorate,decoration={
        brace,
        raise=2mm
    }] (4.6, 6) -- (4.778, 6) node [pos=0.5,anchor=south,yshift=3mm] {$\frac{\epsilon  (1+\epsilon)}{2 w}$};

            \draw [decorate,decoration={
        brace,
        mirror,
        raise=2mm
    }] (0, 0) -- (4.778, 0) node [pos=0.5,anchor=north,yshift=-3mm] {$\frac{w^2 + \epsilon - 2c\epsilon}{2w}$}; 
    
      \draw [decorate,decoration={
        brace,
        mirror,
        raise=2mm
    }] (4.8, 0) -- (6, 0) node [pos=0.5,anchor=north,yshift=-3mm] {$\frac{(1-2 c) (1-\epsilon)}{2 w}$}; 
\end{tikzpicture}
\end{center}

First, the $m\ell$ bisector is the line $y = (1-\epsilon) / 2$. 

Meanwhile, the $\ell r$ bisector has slope $ -\frac{w}{1-2c}$ and passes through $(w/2, 1/2)$, so can be written as $y = 1/2 -\frac{w}{1-2c}(x - w/2)$. Plugging in $y = (1-\epsilon)/2$ yields the $x$ coordinate of the dotted line, $x = \frac{w^2 + \epsilon - 2c\epsilon}{2w}$. Plugging $y = 0$ and subtracting the dotted line's $x$ coordinate, we find the width of lower triangle to be $\frac{(1-2 c) (1-\epsilon)}{2 w}$. 

Next, the $mr$ bisector has slope $-\frac{w}{\epsilon}$ and passes through $(w/2, 1 - c - \epsilon /2)$, so it can be written as $y  = 1 - c - \epsilon / 2 - \frac{w}{\epsilon}(x - w/2) $. Plugging in $y = 1$, we find the intersection with the top of the rectangle to be at $x =  \frac{w^2 -2 c \epsilon -\epsilon ^2}{2 w}. $ Subtracting this from the dotted line's $x$ coordinate yields the width of the top triangle, $ \frac{\epsilon  (1+\epsilon)}{2 w}$.

The vote shares are thus:
\begin{align*}
  v(m)   &= \frac{w^2 + \epsilon - 2c\epsilon}{2w} \cdot \frac{1+\epsilon}{2} - \frac{\epsilon  (1+\epsilon)^2}{8 w} \\
  &= \frac{(1+\epsilon) \left(2 w^2-4 c \epsilon-\epsilon ^2+\epsilon \right)}{8 w}\\[0.5em]
  v(\ell)   &= \frac{w^2 + \epsilon - 2c\epsilon}{2w} \cdot \frac{1-\epsilon}{2} + \frac{(1-2 c) (1-\epsilon)^2}{8 w} \\
  &= \frac{(1-\epsilon) \left(2 w^2-2 c (\epsilon +1)+\epsilon +1\right)}{8 w}\\[0.5em]
    v(r)   &= \frac{w^2 - \epsilon + 2c\epsilon}{2w} - \frac{(1-2 c) (1-\epsilon)^2}{8 w} + \frac{\epsilon  (1+\epsilon)^2}{8 w}\\
    &= \frac{4 w^2 + (2 c - 1 + \epsilon) (1+\epsilon)^2}{8 w}.
\end{align*}
Consider $v(\ell) - v(m)$; since we only care about the sign, we can multiply both by $8w$:
\begin{align*}
  8w[v(\ell) - v(m)] &= (1-\epsilon) \left(2 w^2-2 c (\epsilon +1)+\epsilon +1\right) - (1+\epsilon) \left(2 w^2-4 c \epsilon-\epsilon ^2+\epsilon \right)\\
  &=1 + 6 c \epsilon ^2+4 c \epsilon -2 c-4 w^2 \epsilon +\epsilon ^3-\epsilon ^2-\epsilon\\
  &\ge 1  -2 c-4 w^2 \epsilon -\epsilon ^2-\epsilon\\
  &> 1  -2 c-(4 w^2+2) \epsilon \tag{since $\epsilon < 1$}\\
  &> 1 - 2c -(4w^2 + 2) \frac{1 - 2c}{4w^2 + 2}\tag{since $\epsilon < \frac{1 - 2c}{4w^2 + 2}$}\\
  &= 0.
  \end{align*}
  Thus $m$ has a smaller vote share than $\ell$.

Now consider the sign of $v(r) - v(m)$, again multiplying for simplification:
\begin{align*}
  8w[v(r) - v(m)] &= 4 w^2 + (2 c - 1 + \epsilon) (1+\epsilon)^2 - (1+\epsilon) \left(2 w^2-4 c \epsilon-\epsilon ^2+\epsilon \right)\\
  &= 6 c \epsilon ^2+8 c \epsilon +2 c-2 w^2 \epsilon +2 w^2+2 \epsilon ^3+\epsilon ^2-2 \epsilon -1\\
  &\ge -2 w^2 \epsilon +2 w^2 -2 \epsilon -1\\
  &> -2w^2 / 6 + 2w^2 - 1/3 - 1 \tag{since $w > 1$ and $\epsilon <\frac{1-2c}{4w^2 + 2}$, $\epsilon < 1/6$}\\
  &= \frac{5}{3}w^2 - \frac{4}{3} > 0. \tag{since $w > 1$}
\end{align*}
Thus $m$ also has a smaller vote share than $r$ and $m$ is eliminated first. By symmetry of the remaining candidates $\ell$ and $r$, $r$ is a possible IRV winner.
\end{proof}

\begin{lemma}\label{lemma:rect-minimal-l2}
  In any rectangle $R = [0, w] \times [0, h]$ with uniform $L_2$ voters, there is a sequence of candidate configurations satisfying \Cref{lemma:chain}.
\end{lemma}
\begin{proof}
WLOG suppose $w \ge h$. We can rescale the rectangle so that $h = 1$ without affecting elections, so suppose WLOG that $h = 1$. The case $w = 1$ is already known from the square analysis, so suppose $w > 1$. 

Consider the 3-candidate configuration with candidates at $c = (x, 1/2)$, $t = (x - \epsilon, 2/3)$, $b = (x - \epsilon, 1/3)$ with $\epsilon = \frac{1}{80 w}$. The $ct$ bisector passes through $(x - \epsilon / 2, 7/12)$ and has slope $6\epsilon$, so it intersects the right side of the rectangle no higher than (upper bounding by setting $x = \epsilon / 2$):
\begin{align*}
  \frac{7}{12} + w\cdot  6\epsilon &= \frac{7}{12} + \frac{6w}{80w}\\
   &= \frac{79}{120} < \frac 2 3.
\end{align*}
By symmetry, the $bc$ bisector intersects the right side of the rectangle at a point above $1/3$. Thus, $c$'s vote share is strictly less than $1/3$ and $c$ is first eliminated (angles exaggerated in the visualization for clarity; construction has subtler angles so that even when $c$ is close to the left edge, $c$'s vote share is contained within the dotted lines):

 \begin{center}
 \begin{tikzpicture}[x=6mm,y=6mm]
 
  \draw[draw=black, thick] (0, 0) rectangle ++(10, 6);

  \draw[dashed] (0, 3.2018) -- (4.97, 3.5) -- (10, 3.8036);
  \draw[dashed] (0, 2.7982) -- (4.97, 2.5) -- (10, 2.1964);
  
  \draw[dotted] (0, 2) -- (10, 2);
  \draw[dotted] (0, 4) -- (10, 4);

  \node[fill, circle,inner sep=2pt,label=$t$] (B) at (4.94, 4) {};
  \node[fill=red, circle,inner sep=2pt, label=0:$c$] (C) at (5, 3) {};
  \node[fill, circle,inner sep=2pt, label=270:$b$] (T) at (4.94, 2) {};
\end{tikzpicture}
  \end{center}
Just as in the proof of \Cref{lemma:rect-minimal-l1}, we can thus make progress towards the left edge, as either $t$ or $b$ wins in the above configuration (for any $x \in [\epsilon, w/2]$). Thus, let $C_1, \dots, C_{\lceil80w \rceil}$ be the sequence of configurations where $C_i$ is the three-candidate configuration above with $x = w/2 - (i-1)/2 \epsilon$ for odd $i$ and the two candidate configuration $\{(w/2 - i \epsilon / 2, 1/2)$, $(1/2 - i \epsilon/2, 7/12)\}$ for even $i$. If $w$ is not an integer and the last configuration would make $x < \epsilon$, we can just make $\epsilon$ smaller in the final step so the last configuration ends at $x$-coordinate $0$.   

Next, we use the configuration $C_{n+1}$ (letting $n = \lceil80w \rceil$) with three candidates $(0, 1/6)$, $(0, 1/2)$,and $(0, 5/6)$, with a possible winner at $(0, 5/6)$ (the three have tied vote shares $w/3$). We can then begin using the construction from \Cref{lemma:rectangle-l2-corner-squeeze}, with $\epsilon = \frac{1}{8w^2 + 5}$ and $c = 1/6 - \epsilon$; note that this satisfies the requirements on $\epsilon$ and $c$ from \Cref{lemma:rectangle-l2-corner-squeeze}. This gives a configuration with a candidate at $(0, 5/6)$ (maintaining the chain) and a possible winner at $(w, 1-c)$. To get back to the left side of the rectangle, we can use the two-candidate configuration $C_{n+2} =  \{(0, 1-c), (w, 1-c)\}$, with a possible winner at $(0, 1-c)$ by symmetry. In the configurations $C_{n+1}, C_{n+2}$ we have thus moved the possible winner from $(0, 5/6)$ to $(0, 5/6 + \frac{1}{8w^2 + 5})$. We then repeat this pair of configurations with the same $\epsilon$, but with $c = 1/6 - i \epsilon$, for $i = 2, \dots, \lceil \frac{8w^2 + 5}{6} \rceil$   $O(w^2)$, each time making progress towards the top left corner. In the final iteration we may overshoot, so we can make $c = 1$. This finally gives us a possible winner at a corner of the rectangle (after $O(w^2)$ steps), which is a Condorcet loser by \Cref{lemma:hyperrect-condorcet}. Thus, the sequence of candidate configurations we have constructed satisfies \Cref{lemma:chain}: we began at the Condorcet winner, ended at a Condorcet loser, and a possible winner at each step lost in the next.

\end{proof}

\begin{proof}[Proof of \Cref{thm:hyperrect-trivial-minimal}]
By \Cref{lemma:rect-to-hyper}, it suffices to show that for any rectangle with uniform $L_1$ or $L_2$ voters, there is a sequence of configurations satisfying \Cref{lemma:chain}. This is precisely what \Cref{lemma:rect-minimal-l1} and \Cref{lemma:rect-minimal-l2} give.
\end{proof}

%\begin{proof}[Proof of \Cref{thm:barbell}]
%First, the total area of $B$ is $1 + 8/10 + 4 = 5.8$. Suppose there is only one candidate $r$ remaining in $S$. We will show that any voter $v$ in the larger square above $y = 1/10$ votes for $r$ (making their vote share at least $3.8$, so they win). If $r$ is in the larger square, we are done: the largest $L_1$ distance between two points in the square is $4$, while every point outside of $S$ is distance $>4$ from the larger square. If $r$ is in the pipe between the squares, then consider a shortest path from $v$ to $r$ which first goes to the left of the large square, then drops down to the top of the pipe, then travels left to $r$'s $x$ coordinate, and then drops down to $r$. Also consider a shortest path to any candidate not in $S$, which starts the same way, but continues traveling left beyond $r$'s $x$ coordinate. Due to the slant at the left boundary of $S$, such a shortest path will necessarily travel further to leave $S$ than to reach $r$.
%\end{proof}

\subsection{Voting on graphs}

\begin{proof}[Proof of \Cref{thm:exclusion-co-np}]
  If $S$ is not an IRV exclusion zone, then there exists some configuration of candidates with at least one in $S$ where a candidate outside of $S$ wins (perhaps depending on tie-breaks). In such a case, a co-NP verifier can provide the counterexample (including the tie-breaking order), which we can check in polynomial time by running IRV.
  
  We now show NP-hardness by reduction from (co)-RX3C (restricted exact cover by 3-sets, where every item appears in exactly 3 sets, which is NP-complete)~\cite{gonzalez1985clustering}. Let $(X, C)$ be an instance of RX3C, where $X$ is a set of $3n$ items for some integer $n$ and $C = \{C_1, \dots, C_{3n}\}$ is a collection of size-3 subsets of $X$ (since every item appears in exactly three 3-sets, the total number of 3-sets is also $3n$). 
   Construct an instance of \textsc{IRV-Exclusion} $(G, S)$ as follows.
  
  Nodes of $G$:
  \begin{itemize}
    \item $(5n-1)/2$ \emph{item nodes} $x_1, \dots, x_{(5n-1)/2}$ for each item $x \in X$,
    \item $9n$ \emph{set-item nodes} $y_{i, x}$ for each set $C_i$ and each $x \in C_i$
    \item one \emph{set node} $c_i$ for each set $C_i \in C$,
    \item two \emph{winning nodes} $s_1$ and $s_2$,
    \item $3n+2$ \emph{bonus nodes} $b_1, \dots, b_{k+2}$.
  \end{itemize}
  
  Edges of $G$:
  \begin{itemize}
    \item Set nodes connect to their item nodes and their set-item nodes: for each $c_i$, the edges $(c_i, x_1), \dots, (c_i, x_{(5n-1)/2})$ and $(c_i, y_{i, x})$ for each $x \in C_i$,
    \item A clique among the set and winning nodes: edges connecting all pairs of nodes among $\{c_1, \dots, c_k, s_1, s_2\}$,
    \item Set nodes each connect to their bonus node: edges $(c_i, b_i)$ for each $i = 1\dots k$,
    \item Winning nodes each connect to their bonus node: the edges $(s_1, b_{k+1})$ and $(s_2, b_{k+2})$,
    \item Winning nodes connect to every set-item node: the edges $(s_1, y_{i, x})$ and $(s_2, y_{i, x})$ for each $C_i\in C$ and each $x \in C$.
  \end{itemize}  
  Claim: $S = \{s_1, s_2\}$ is an IRV exclusion zone of $G$ if and only if $(X, C)$ does not have an exact cover by 3-sets.
  
  Suppose $(X, C)$ has an exact cover by 3-sets, $C'\subseteq C$ (which must have exactly $n$ sets). Consider the configuration with a one candidate at $s_1$ and one in $c_i$ for each $C_i \in C'$ (of which there are $n$). Nodes vote as follows:
  \begin{itemize}
    \item bonus nodes connected to a candidate vote for the candidate,
    \item bonus nodes not connected to a candidate are equally shared among all candidates (they are all 2 hops away from every candidate), same for non-occupied $c_i$ nodes (they are 1 hop from every candidate), and for $s_2$ and $b_{k+2}$; call this total evenly-split vote share $c$,
    \item $x$ nodes vote for the $c_i$ covering them (all $(5n-1)/2$ copies),
    \item the $3n$ $y$ nodes corresponding to set-item pairs in the exact cover split their vote between their $c_i$ and $s_1$,
    \item the other $6n$ $y$ nodes vote fully for $s_1$.
  \end{itemize}
  Thus, the total vote share of $s_1$ is
  \begin{align*}
    \underbrace{3n/2}_{\text{split $y$ nodes}} + \underbrace{6n}_{\text{owned $y$ nodes}} + \underbrace{1}_{s_1} + \underbrace{1}_{b_{k+1}} + \underbrace{c/(n+1)}_{\text{evenly-split votes}} = 15n/2 + 2 + c/(n+1)
  \end{align*}
  Meanwhile, the total vote share of each occupied $c_i$ node is:
  \begin{align*}
    \underbrace{3 \cdot (5n-1)/2}_{\text{$x$ nodes}} +\underbrace{3/2}_{\text{split $y$ nodes}} + \underbrace{1}_{c_i} + \underbrace{1}_{b_i} + \underbrace{c/(n+1)}_{\text{evenly-split votes}} = 15n/2 + 2 + c/(n+1)
  \end{align*}
  Thus, $s_1$ can be eliminated next by a tie-breaker, and $\{s_1, s_2\}$ is not an IRV exclusion zone.
  
Now suppose $(X, C)$ does not have an exact cover by 3-sets. We want to show that $\{s_1, s_2\}$ is an IRV exclusion zone. So, consider any configuration with at least one candidate from $\{s_1, s_2\}$. At some point, only one of them will remain; since they have the exact same connectivity, we can call the last remaining node $s_1$ WLOG. First, we'll show that $s_1$ is guaranteed to outlive every $x$ node, every $y$ node, and every bonus node. First, consider the bonus node $b_{k+1}$ attached to $s_1$. If occupied, it has vote share 1, but $s_1$ has vote share $1$ plus a share of $s_2$ (which we know is unoccupied), so $s_1$ outlives $b_{k+1}$ and then gets vote share 2 (plus its share of other unoccupied nodes). Any other bonus node then has smaller vote share than $s_1$, since they only get vote share 1 plus a share of the unoccupied $c_i$ nodes (which they are equally sharing with $s_1$). Thus, at some point all of the bonus nodes are eliminated, with $s_1$ still in contention. Item nodes (both $x$ and $y$) also only get vote share 1 plus a share of adjacent unoccupied $c_1$ nodes (and a share of $s_2$, for $y$ nodes), but these are also evenly split with $s_1$. Thus every $x$ and $y$ node is eliminated before $s_1$. We therefore only need to worry about what happens when only $s_1$ and $c_i$ nodes remain. As we saw, if there is an exact cover by 3-sets, candidates occupying those nodes will exactly tie $s_1$. But now we know no such cover exists. Thus, regardless of which $c_i$ nodes are occupied, either: (1) there is at least one uncovered element; or (2), there is at least one multiply-covered element. As we will see, this will cause inefficiency in the vote counts for set nodes that tips the balance in favor of winning nodes.

Suppose there are $k$ set nodes remaining. Let $\ell \le 3k$ be the number of covered items and let $v'$ denote votes from only $x$ and $y$ nodes (note that the rest of the votes are equal: one self-vote, one vote from connected bonus node, and the evenly-split votes as above).  Then $s_1$'s vote share from $y$ nodes is $3k/2 + (9n - 3k) = 9n - 3k/2$---it depends only on the number of set nodes remaining. Note that $s_1$ may also get votes from $x$ nodes that are uncovered, but these will be equally shared with set nodes and can be ignored; regardless, we have $v'(s_1) \ge 9n - 3k/2$.

 We'll show by cases for $k$ that some $c_i$ node always has fewer votes than $s_1$ (when there is no exact cover). To do this, we'll use the following lemma.
 \begin{lemma}\label{lemma:average-coverage}
 Suppose there are $k> cn$ occupied set nodes. Define the item coverage counts of a set $C_i$ to be the number of times each item $x\in C_i$ is covered. There must be some set node whose sum of item coverage counts is $> 3c$. 
 \end{lemma}
 \begin{proof}
Since there are $k > cn$ occupied set nodes, there are $3k > 3cn$ items in all occupied sets. There are only $3n$ items, so on average, each item is covered $> c$ times (call this the average coverage). To get an upper bound on average coverage, we can average over the item coverage counts of the occupied sets (this will overcount items covered by multiple sets, hence the upper bound). Suppose no set has a sum of coverage counts $> 3c$, so the sum of all set coverage counts is $\le 3kc$. Averaging over items, we get an upper bound on average coverage of $\le c$ (since there are $3k$ items total across occupied sets), contradicting that the average coverage is $> c$. Thus some set has sum of coverage counts $> 3c$. 
 \end{proof}
 
 We can now proceed by cases on $k$, the number of set nodes remaining. In each case, we will find some $c_i$ with $v'(c_i) < 9n-3k/2 \le v'(s_1)$, which is eliminated before $s_1$.
 \begin{enumerate}
  \item $k < n$. Then $v'(s_1) > 9n - 3n/2 = 15n/2$. The total vote share $c_i$ nodes get from $x$ and $y$ nodes is $3k/2 + \ell (5n-1)/2 \le 3k/2 + 3k(5n-1)/2$. Thus, the average $v'(c_i)$ is at most $3/2 + 3(5n-1)/2 = 15n/2$, so by the pigeonhole principle some node $c_i$ has $v'(c_i) \le 15n/2 < 15n/2 + 3n/2 - 3k/2 = 9n-3k/2 \le v'(s_1)$.
  \item $k = n$. Since there is no exact cover, we must have $\ell < 3k$. The pigeonhole argument is the same as above, except now we find that there must be some $c_i$ with $v'(c_i) < 15n/2 = 15n/2 + 3n/2 - 3k/2 = 9n - 3k/2 \le v'(s_1)$. 
  \item $n < k \le 5n/3$. Then $v'(s_1) \ge 9n - 3(5n/3)/2 = 6.5n$. Each item is covered $> 1$ times on average, so by \Cref{lemma:average-coverage} there must be a set whose coverage counts are element-wise at least $(2, 1, 1)$. This set has $v'(c_i) \le 3/2 + 2(5n-1)/2 + (5n-1)/4 = 6.25n + 1/4 <  6.5n = 9n - 3(5n/3)/2 \le v'(s_1)$.
  \item $5n/3 < k \le 8n/3$. Then $v'(s_1) \ge 9n - 3(8n/3)/2 = 5n$. Each item is covered $> 5/3$ times on average, so by \Cref{lemma:average-coverage} there must be a set whose coverage counts are element-wise at least $(2, 2, 2)$ or at least $(3, 2, 1)$. In the first case, this set has $v'(c_i) \le 3/2 + 3(5n-1)/4 = 3.75n + 3/4 < 5n =  9n - 3(8n/3)/2 \le  v'(s_1)$. In the second case, this set has $v'(c_i) \le 3/2 + (5n-1)/2 + (5n-1)/4 + (5n-1)/6 = 4.58\overline{3}n + 7/12  <  v'(s_1)$.
  \item $8n/3 < k \le 3n$. Then $v'(s_1) \ge 9n - 3(3n)/2 = 4.5n$. Each item is covered $> 8/3$ times on average, so  by \Cref{lemma:average-coverage} there must be a set whose coverage counts are $(3, 3, 3)$. This set has $v'(c_i) = 3/2 + 3(5n-1)/6 = 2.5n + 1 < 4.5n = 9n - 3(3n)/2 \le v'(s_1)$. 
\end{enumerate}
Regardless of $k$, we thus find that some $c_i$ node has a smaller vote share than $s_1$ (when there is no exact cover), so $s_1$ cannot be eliminated and is the IRV winner. Thus, $\{s_1, s_2\}$ is an IRV exclusion zone.

This construction has size polynomial in the size of the RX3C instance. Thus, \textsc{IRV-Exclusion} is co-NP-complete. 
\end{proof}

\begin{proof}[Proof of \Cref{thm:approx}]

%Given a graph $G$, suppose we want to find some (hopefully small) set of nodes which is probably approximately an IRV exclusion zone. If we place a candidate at every node, then the set of possible winners must all be in the minimal exclusion zone. Unfortunately, determining whether some candidate can win under tie-breaking is NP-complete in general~\cite{conitzer2009preference}. But we can do a number of realizations to get some starting set of possible winners (which may be incomplete).

First, we show the first claim, that the algorithm returns a set $S$ which is a subset of the minimal exclusion zone. Any winner in the configuration with a candidate at every node must be in the minimal IRV exclusion zone, or we would have a counterexample. Thus $S$ is initialized to a subset of the minimal IRV exclusion zone. By \Cref{prop:exclusion-facts}, if $u$ is in the minimal exclusion zone, then any node in $d_L(u)$ must be too (since a path in the loss graph gives us a sequence of configurations to which we can apply \Cref{prop:exclusion-facts}). Moreover, whenever we add a node $w$ to $S$ in step 4(b), it must also be in the minimal exclusion zone, since otherwise the configuration would be a counterexample. Thus all the of nodes we add to $S$ during the algorithm must be in the minimal IRV exclusion zone. 

  Now we show the second claim. Let $X_S$ be a random variable taking value $1$ if a uniformly random configuration of candidates with no clones and at least one candidate in $S$ has IRV winner (with uniform tie-breaking) outside of $S$ and $0$ otherwise. By definition, $S$ is an $\epsilon$-approximate IRV exclusion zone iff $E[X_S] < \epsilon$. By Hoeffding's inequality,
  
  \begin{align*}
    \Pr\left(\left|\frac{1}{t}\sum_{i = 1}^{t} (X_S - E[X_S]) \right|\ge \epsilon \right) \le 2\exp(- 2 t\epsilon^2 )
  \end{align*}
  
  Since we observe $t$ independent draws where $X_S = 0$ (for the $S$ returned by the algorithm), the probability that $E[X_S]\ge \epsilon$ is at most  $2\exp(- 2 t \epsilon^2 )$. We want $2\exp(- 2 t \epsilon^2 ) \le \delta $, so solving for $t$ yields $t \ge \log(2/\delta)/(2 \epsilon^2)$. This is precisely how we picked the number of iterations of the algorithm, satisfying this inequality. Thus, the probability that the $S$ returned by the algorithm is a $(1-\epsilon)$-approximate IRV exclusion zone is at least $1-\delta$. 
  
  Finally, we show runtime. There are $O(n^2)$ pairwise contests. Each contest can be resolved by a single BFS pass taking $O(n + m)$ time, so it takes $O(n^3 + n^2m)$ time to construct $L$ in step 1, which has $n$ nodes and $O(n^2)$ edges. Running $IRV$ on $G$ with a candidate at every node takes $O(n^2+nm)$ time, since we need to run $n$ BFS passes (one for each elimination). Thus the total runtime of step 2 is $O(n^3 + n^2 m)$. We can run step 3 in time $O(n^3)$ by performing a BFS from each node in $L$ to find $d_L(u)$ (this takes $O(n^3)$ time since we run $n$ searches on a graph with $O(n^2)$ edges). We iterate at most $O(n\log(1/\delta)/\epsilon^2)$ times in step 4, since every $\log(1/\delta)/\epsilon^2$ steps we either terminate or add a node to $S$. Each iteration takes $O(n^2 + nm)$ time for up to $n$ BFS passes while running IRV. Combining the four steps makes the overall running time $O((n^3 +n^2m)\log(1/\delta)/\epsilon^2)$.

\end{proof}
\end{document}